\documentclass[onefignum,onetabnum]{siamonline220329}

\usepackage{lipsum}
\usepackage{amsfonts,amsmath,amssymb}
\usepackage{graphicx}
\usepackage{epstopdf}
\usepackage{algorithmic}
\usepackage{nicefrac}
\usepackage{microtype}
\usepackage{booktabs}
\usepackage{tikz}
\usepackage{pgfplots}
\usepackage{tikz-3dplot}
\usepackage{siunitx}
\usepackage{subfig}
\usetikzlibrary{3d}
\ifpdf
  \DeclareGraphicsExtensions{.eps,.pdf,.png,.jpg}
\else
  \DeclareGraphicsExtensions{.eps}
\fi
 \pgfplotsset{compat=1.18}
\usepackage{enumitem}
\setlist[enumerate]{leftmargin=.5in}
\setlist[itemize]{leftmargin=.5in}

\newsiamremark{remark}{Remark}
\newsiamremark{hypothesis}{Hypothesis}
\crefname{hypothesis}{Hypothesis}{Hypotheses}
\newsiamthm{claim}{Claim}
\newsiamthm{example}{Example}

\headers{Enforcing Katz and PageRank}{S. Cipolla, F. Durastante, B. Meini}

\title{Enforcing Katz and PageRank Centrality Measures in Complex Networks\thanks{Submitted to the editors \today.
\funding{This work has been partially supported by:  Spoke 1 ``FutureHPC \& BigData''  of the
Italian Research Center on High-Performance Computing, Big Data and Quantum%
Computing (ICSC)  funded by MUR Missione 4 Componente 2 Investimento 1.4:%
Potenziamento strutture di ricerca e creazione di ``campioni nazionali di R\&S 
(M4C2-19)'' - Next Generation EU (NGEU); by the 
``INdAM – GNCS Project: Metodi di riduzione di modello ed approssimazioni di rango basso per problemi alto-dimensionali'' code CUP\_E53C23001670001; and by the PRIN project ``Low-rank Structures and Numerical Methods in Matrix and Tensor Computations and their Application'' code 20227PCCKZ. The authors are member of the INdAM GNCS group.}}}

\author{Stefano Cipolla\thanks{School of Mathematical Sciences, University of Southampton, Southampton, UK
  (\email{s.cipolla@soton.ac.uk}, \url{stefanocipolla.github.io}).}
\and Fabio Durastante\thanks{Department of Mathematics, University of Pisa, Pisa, IT 
  (\email{fabio.durastante@unipi.it}, \url{fdurastante.github.io}).}
\and Beatrice Meini\thanks{Department of Mathematics, University of Pisa, Pisa, IT (\email{beatrice.meini@unipi.it}).}}

\usepackage{amsopn}

\usepackage{listings}
\usepackage{matlab-prettifier}

\ifpdf
\hypersetup{
  pdftitle={Enforcing Katz and PageRank Centrality Measures in Complex Networks},
  pdfauthor={Stefano Cipolla, Fabio Durastante, Beatrice Meini}
}
\fi

\usepackage{multirow}
\usepackage{lscape}

\begin{document}

\maketitle

\begin{abstract}
We investigate the problem of  enforcing a desired centrality measure in complex networks, while still keeping the original pattern of the network. Specifically, by representing the network as a graph with suitable nodes and weighted edges, we focus on computing the smallest perturbation on the weights required to obtain a prescribed PageRank or Katz centrality index for the nodes. 
Our approach relies on optimization procedures that scale with the number of modified edges, enabling the exploration of different scenarios and altering network structure and dynamics.
\end{abstract}



\section{Introduction}\label{sec:introduction}
A centrality measure  quantifies the relative importance of a node in a complex network. It involves computing a centrality score for each node using various definitions of importance, among the most widely used are the Katz~\cite{Katz195339}, total communicability~\cite{10.1093/comnet/cnt007}, eigenvector~\cite{landau1895relativen,Schoenberg1969}, PageRank~\cite{Page1998107}, and matrix-function based centrality indexes~\cite{MR2736969}. These scores represent the node's influence within the network, assuming that nodes with higher centrality scores are considered more central or influential. Network analysis applications of centrality measures include identifying central decision-makers, analyzing network dynamics, and studying the spread of information or epidemics.

In this work, we analyse the problem of modifying
the connections of a given complex network to produce a perturbed network having assigned target Katz or PageRank centralities of the nodes.
{For example, enhancing a specific node's Katz or PageRank centrality may be desirable to increase its influence within the network or to reduce the centrality of the most dominant nodes.}
This could be achieved by promoting connections to or from suitable nodes or by altering their attributes. 
However, it would be desirable to modify the network not  to alter its topology and to be as small as possible.  

Problems concerning the modification, and often the optimization, of several network measures have already been addressed in the literature. For example, in~\cite{MR3439773},  updating and downdating strategies are proposed to optimize the communicability of the network. The problem of minimizing the spectral radius of the adjacency matrix  of the underlying network by either low-rank updates or optimization procedures has been addressed in~\cite{doi:10.1137/1.9781611974010.64,PhysRevE.84.016101,7373366}. 
In~\cite{massei2023optimizing} the problem of attaining the maximal increase/reduction of a global robustness measure of a complex network through edge weight modifications is studied. 
{The problem of modifying the adjacency matrix of a graph or a stochastic matrix, in order to have a given  Perron vector,  is addressed in~\cite{benzi2025,berkhout-heidergott-vandooren,  Chan2016, 10.1145/3018661.3018703,gillis2024assigning, Nicosia2012}. Our work is primarily inspired by \cite{gillis2024assigning}, where the authors seek the minimal perturbation of a stochastic matrix’s entries needed to obtain a given stationary distribution while preserving its sparsity structure. Additionally, a closely related problem—determining the smallest perturbation to a graph’s adjacency matrix to achieve a specified eigenvector centrality—has been examined in~\cite{benzi2025}.}

Our objective is to determine the minimum norm perturbation of the adjacency matrix of a graph to enforce a desired  centrality, where the centrality is either the Katz or the PageRank measure. Moreover, we investigate also the case in which we are required to perturb only a subset of the existing edges to reduce the modification to the network's topology. Indeed, for the Katz centrality we show that, given a graph and a vector that represents the desired measure of all the nodes, we may find a modification of a subset of the existing edges so that the nodes of the modified network have the desired centrality.
{For the PageRank centrality, the desired measure can be achieved by modifying existing edges and, if necessary, by adding self-loops, as in the approach in~\cite{gillis2024assigning}. For both the Katz and the PageRank  measure, we seek, among all admissible perturbations, one that minimizes a convex combination of the Frobenius norm and the 1-norm. This approach contrasts with that in~\cite{gillis2024assigning}, which has a single norm as objective; here, the Frobenius norm quantifies the average magnitude of the perturbation, while the 1-norm promotes sparsity~\cite{MR1379242}.
Although the PageRank measure can be interpreted as the stationary distribution of an appropriately defined stochastic matrix, our method differs from that of~\cite{gillis2024assigning} in that we modify the underlying graph structure rather than the stochastic matrix in the PageRank model. In the context of PageRank, our methodology is more closely related to the \emph{PageRank Optimization Problem}~\cite{MR3175058,MR3006713}, where the goal is to maximize (or minimize) the PageRank of a node by strategically adding or removing links from a given subset.}

{The desired perturbation of the adjacency matrix, both for the Katz and for the PageRank measure, is expressed as the solution of a suitable Quadratic Programming (QP) problem.}
In the two cases, the specific structure of the problem is exploited in the {numerical} solution of the QP problem by the Interior Point Method. 
We use scalable optimization procedures, and our algorithm code and examples are available in a dedicated \texttt{GitHub} repository.

The paper is organized as follows. Section~\ref{sec:prel} provides basic definitions of complex networks and centrality measures. Sections~\ref{sec:katz} and~\ref{sec:pagerank} analyze problems related to Katz and PageRank centrality measures, respectively. 
Section~\ref{sec:numerics} presents numerical tests validating our strategies on various real networks from both modeling and algorithmic perspectives. Finally, Section~\ref{sec:conclusions} concludes and suggests future work on other centrality measures.

\subsection{Notation} Given any matrix $A \in \mathbb{R}^{n \times n}$ we denote by 
\[
\Omega_A = \{(i,j) \hbox{ s.t. } 1\leq i,j \leq n,\; [A]_{i,j} \neq 0 \}  \hbox{ and }  {s_{A}=|\Omega_A|},
\]
and by $\mathbb{S}(A)$ the set of matrices
\[
    \mathbb{S}(A) = \{ B \in \mathbb{R}^{n \times n} \,\text{ s.t. } [B]_{i,j} \neq 0 \,\Rightarrow\, (i,j) \in \Omega_A \}.
\]
Furthermore, we denote by $\|\cdot\|_F$ the Frobenius matrix norm, by $\|\cdot\|_1$ the $1$-norm for vectors together with the corresponding induced matrix norm, 
and by $\rho(A)$ the spectral radius of the matrix $A$. We also denote with $\mathbf{1}$ the vector of all ones, with $D_{\mathbf{x}} = \operatorname{diag}(\mathbf{x})$ the diagonal matrix having $\mathbf{x} \in \mathbb{R}^{n}$ on the main diagonal, with $\operatorname{vec}(\cdot)$ the column-wise vectorization of a matrix, while Kronecker and Hadamard products are denoted by the usual ``$\otimes$'' and ``$\circ$'' symbols. To denote the vertical concatenation of vectors we employ the notation $(\mathbf{x};\mathbf{y})$ while horizontal concatenation is denoted as $(\mathbf{x}^{\top},\mathbf{y}^{\top})$.
A real matrix $A$ is said to be nonnegative if $[A]_{i,j}\ge 0$ for any $i,j$. A nonnegative matrix is stochastic if $A\mathbf{1}=\mathbf{1}$.
If $A$ is an irreducible stochastic matrix, the unique vector $\boldsymbol{\pi}$ such that $\boldsymbol{\pi}^{\top}=\boldsymbol{\pi}^{\top}A$ and $\boldsymbol{\pi}^{\top}\mathbf{1}=1$ is called the stationary vector of $A$.

\section{Preliminaries on complex networks and centralities}\label{sec:prel} 
A directed graph is a pair $\mathcal{G} = (V,E)$, where
$V = \{v_1, \ldots, v_n\}$ is a set of $n = |V|$ nodes (or vertices), and $E \subseteq V \times V$ is a set of ordered pairs of nodes
called edges. A weighted directed graph $\mathcal{G} = (V, E, w)$ is obtained by considering a nonnegative weight function $w:V\times V \rightarrow \mathbb{R}^+$ such that $w(v_i,v_j) = w_{i,j} > 0$ if and only
if $(v_i, v_j)$ is an edge of $\mathcal{G}$. For every node $v \in V$, the degree $\deg(v)$ of $v$ is the number of edges insisting on $v$, taking into account their weights:
\begin{equation}\label{eq:definition_degree}
d_i = \deg(v_i) = \sum_{j \, : \, (v_i, v_j) \in E} w_{i,j}.
\end{equation}
The adjacency matrix $A$ is the $n \times n$ matrix with elements
\begin{equation*}
[A]_{i,j} = a_{i,j} = w(v_i,v_j).
\end{equation*}
The degree matrix $D$ is the diagonal matrix whose entries are given by the degrees of the nodes, i.e.,
\begin{equation}\label{eq:definition_degree_matrix}
D  = \operatorname{diag}(d_1, \ldots, d_n) = \operatorname{diag}(A\mathbf{1}).
\end{equation}

If the ordering of the vertices in the edges in $E$ is not relevant,
i.e., if each edge can be traversed both ways, we move from directed
graphs to undirected graphs. An undirected graph is a pair $\mathcal{G} = (V, E)$, where the  set of edges $E$ is such that if $(v_i, v_j) \in E$, then $(v_j, v_i) \in E$ for all $i, j$. A weighted undirected graph $\mathcal{G} = (V, E, w)$ is then obtained by considering a symmetric weight function $w$ with nonnegative values $w(v_i,v_j) \geq 0$ such that $w(v_i,v_j) > 0$ if and only if $(v_i, v_j)$ is an edge of $\mathcal{G}$. The adjacency matrix $A$ of an undirected graph $\mathcal{G}$ is  symmetric.

For any two nodes $u, v \in V$ in a graph $\mathcal{G} = (V, E)$, a walk from $u$ to $v$ is an ordered sequence of nodes $(v_0, v_1, \ldots, v_k)$ such that $v_0 = u$, $v_k = v$, and $(v_i, v_{i+1}) \in E$ for all $i = 0, \ldots, k-1$. The integer $k$ is the length of the walk. An undirected graph $\mathcal{G}$ is connected if for any two distinct nodes $u, v \in V$, there is a walk between $u$ and $v$. A directed graph $\mathcal{G}$ is strongly connected if for any two distinct nodes $u, v \in V$, there is a  walk from $u$ to $v$.

A centrality measure on the graph is a function that assigns to each node, possibly to each edge, an importance value and which allows the nodes, respectively the edges, to be ordered in a ranking. Of the set of available measures, we are interested here in two particular measures: the Katz centrality index~\cite{Katz195339} and PageRank~\cite{MR3376760,Page1998107}. For the sake of readability, we will briefly recall related definitions in the following sections.

\subsection{Katz Centrality}\label{sec:katz-definition}
Let $\mathcal{G} = (V, E)$ be a graph with adjacency matrix $A$.
Given $\alpha>0$ such that 
$\alpha\rho(A)<1$, the Katz centrality~\cite{Katz195339}  of node $v_i$ is the $i$th entry of the vector $\boldsymbol{\mu}=(I-\alpha A)^{-1}\mathbf{1}$. Since $A\ge 0$, then $\boldsymbol{\mu}\ge \mathbf{1}$. We will write that $\mathcal G$ has Katz centrality score $\boldsymbol{\mu}$.

\subsection{PageRank}\label{sec:pagerank-definition}
Let $\mathcal{G} = (V, E)$ be a graph with adjacency matrix $A$.
Given  $\alpha \in (0,1)$ a \emph{teleportation parameter}, and given $\mathbf{t} >\mathbf{0}$ a \emph{personalization} vector such that $\mathbf{t}^{\top}\mathbf{1} = 1$,
the PageRank~\cite{MR3376760,Page1998107} centrality of node $v_i$ is the $i$th entry of the stationary vector $\boldsymbol{\pi}$ of the stochastic matrix $G$, where
\begin{equation}\label{eq:G}
    G = \alpha D^{-1}A +(1-\alpha ) \mathbf{1} \mathbf{t}^{\top}.
\end{equation}
 In other words,
$\boldsymbol{\pi}$
 solves  the equations
\begin{equation}\label{eq:pagerankvector}
 G^\top \boldsymbol{\pi} = \boldsymbol{\pi}, \quad \boldsymbol{\pi}^{\top}\mathbf{1} = 1,
\end{equation}
or, equivalently,
 the nonsingular linear system
\begin{equation}\label{eq:pagerankvector2}
(I-\alpha(D^{-1}A)^{\top} )\boldsymbol{\pi}=(1-\alpha)\mathbf{t}.
\end{equation}
 We will write that $\mathcal G$ has PageRank centrality score $\boldsymbol{\pi}$.

\section{Enforcing Katz Centralities}
\label{sec:katz}
Now suppose we have a graph $\mathcal G$ having
 Katz centrality score $\boldsymbol{\mu}$, as defined in Section~\ref{sec:katz-definition}, for a given parameter {$\alpha>0$ such that $\alpha\rho(A)<1$}. We consider the problem of determining the smallest perturbation $\Delta$ of the adjacency matrix $A$, without modifying its null entries, that allows us to force the Katz centrality to become the vector $\widehat{\boldsymbol{\mu}}\ge \mathbf{1}$. {Given $\beta \in (0, 1]$,} such a problem can be formulated as the following optimization problem
\begin{equation} \label{eq:Katz_1_alpha_beta}
\mathcal{P}^{\text{Katz}}_{\alpha,\beta} \, : \, \begin{array}{rl}
         \displaystyle \min_{\Delta \in \mathbb{S}(A) } & \displaystyle J(\Delta) =\beta \|\Delta\|_F^2 + (1-\beta) \|\Delta \|_1,  \\
        \text{s.t.} & (I - \alpha(A + \Delta))^{-1} \mathbf{1} = \boldsymbol{\widehat{\mu}}, \\
        & A + \Delta \geq 0.
    \end{array}
\end{equation}
The objective function measures the amplitude in the norm of the perturbation. It is obtained as the convex combination of the perturbation amplitude in the Frobenius norm, which measures the amount of weight/energy we must spend to alter the network's connections {and of} the perturbation in norm $\|\cdot\|_1$, which penalizes dense solutions, i.e., tries to localize the perturbation by making it sparse~\cite{MR1379242}.
The constraints, on the other hand, impose that the new perturbed network has the desired centrality vector and that the weights of the edges remain nonnegative.

{The following result proves that when $A\mathbf{1}> \mathbf{0}$, then problem \eqref{eq:Katz_1_alpha_beta} has always a non-empty feasible set.}

\begin{proposition}\label{pro:feasibility-conditions}
Given $\boldsymbol{\widehat{\mu}}\ge \mathbf{1}$,  $A\ge 0$ such that $A\mathbf{1}> \mathbf{0}$, and $\alpha>0$ such that $\alpha\rho(A)<1$, then the set of matrices $\Delta\in\mathbb{S}(A)$ such
that  $(I - \alpha(A + \Delta))\boldsymbol{\widehat{\mu}} - \mathbf{1} = \mathbf{0}$ and
$A+\Delta\ge 0$  is non-empty. Moreover, for any such matrix $\Delta$ we have $\alpha \rho(A+\Delta)<1$.
\end{proposition}

\begin{proof} For each row $i\in\{1,\ldots,n\}$ of $A$, define $\mathbb{S}_i:=\{ j \in \{1,\ldots,n\} \hbox{ s.t. } a_{ij} \neq 0  \}  $.  Since $A\mathbf{1}> \mathbf{0}$, then $\mathbb{S}_i \neq  \emptyset $ for all $i=1,\dots,n$. By definition $a_{ij}=0$ if $j \notin \mathbb{S}_i$.
Then, by choosing the entries $\delta_{ij}$ of $\Delta$ such that
$\delta_{ij}=0$ if $j\not\in \mathbb{S}_i$, and
$a_{ij}+\delta_{ij}=\widehat \mu_j^{-1} \alpha^{-1}|\mathbb{S}_i|^{-1}(\widehat\mu_i-1)$ if $j\in\mathbb{S}_i$, we obtain $A + \Delta\ge 0$ and
\[
\sum_{j=1}^n (a_{ij}+\delta_{ij})\widehat\mu_j=
\sum_{j\in\mathbb{S}_i} (a_{ij}+\delta_{ij})\widehat\mu_j=\sum_{j\in\mathbb{S}_i} \alpha^{-1}|\mathbb{S}_i|^{-1}(\widehat\mu_i-1)   =   \alpha^{-1}(\widehat\mu_i-1),
\]
i.e., $\alpha(A + \Delta)\boldsymbol{\widehat{\mu}} = \boldsymbol{\widehat{\mu}} - \mathbf{1}$. Concerning the spectral radius of $A+\Delta$, since
$\alpha (A+\Delta){\boldsymbol{\widehat{\mu}}} = {\boldsymbol{\widehat{\mu}}} - \mathbf{1}$, we find that
\[\alpha D_{\boldsymbol{\widehat{\mu}}}^{-1} (A+\Delta) D_{\boldsymbol{\widehat{\mu}}} \mathbf{1} = \mathbf{1} - D_{\boldsymbol{\widehat{\mu}}}^{-1} \mathbf{1}.\] Since $A+\Delta\ge 0$, we have
\[
\alpha \| D_{\boldsymbol{\widehat{\mu}}}^{-1} (A+\Delta) D_{\boldsymbol{\widehat{\mu}}} \|_\infty = \| \mathbf{1} - D_{\boldsymbol{\widehat{\mu}}}^{-1} \mathbf{1}\|_\infty=1-\min_j \widehat{\mu}_j^{-1}<1,
\]
 therefore $\rho(A+\Delta)<1/\alpha$, which concludes the proof.
\end{proof}

If $A$ is nonnegative and irreducible, the condition $A \mathbf{1} >\mathbf{0}$ is automatically satisfied.

In light of Proposition~\ref{pro:feasibility-conditions}, we {have that the problem in~\eqref{eq:Katz_1_alpha_beta} is well defined and can be written, hence,} in the equivalent form
\begin{equation}\label{eq:Katz_1_alpha_1}
\mathcal{P}^{\text{Katz}}_{\alpha,\beta} \, : \, \begin{array}{rl}
         \displaystyle \min_{\Delta \in \mathbb{S}(A) } & \displaystyle J(\Delta) = \beta \|\Delta\|_F^2 + (1-\beta) \|\Delta \|_1,  \\
        \text{s.t.} & (I - \alpha(A + \Delta))\boldsymbol{\widehat{\mu}} - \mathbf{1} = \mathbf{0}, \\
        & A + \Delta \geq 0.
    \end{array}
\end{equation}

\begin{remark}\label{remark:generalpattern}
{It is important to note that the} same feasibility result of Proposition~\ref{pro:feasibility-conditions}  holds if $\Delta \in \mathbb{S}(M)$, for any $M \in \mathbb{R}^{n \times n} $ such that $M \geq 0$, $M\mathbf{1} > \mathbf{0}$. {Hence, we are able to characterize the admissible modification patterns for the network so that a centrality vector $\boldsymbol{\widehat{\mu}}$ can be obtained.}
\end{remark}

{The following proposition gives} lower bounds on the norms of the perturbation $\Delta$ as a function of the difference between $\boldsymbol{\widehat{\mu}}$ and $\boldsymbol{{\mu}}$:

\begin{proposition}\label{prop:bound_delta}
Assume that $\boldsymbol{{\mu}} = (I - \alpha A )^{-1} \mathbf{1}$ and
$\boldsymbol{\widehat{\mu}} = (I - \alpha(A + \Delta))^{-1} \mathbf{1}$.  Then
\begin{align}
& \| \Delta \|_F \ge \alpha^{-1} \frac{\| \mathbf{1} - D^{-1}_{\boldsymbol{\widehat{\mu}}} \boldsymbol{\mu} \|_F  }{\| (I - \alpha A )^{-1} \|_F \sqrt{n}} \cdot \frac{\min_i \widehat\mu_i}{\max_i \widehat\mu_i},\label{eq:norm2}\\
& \| \Delta \|_1 \ge \alpha^{-1} \frac{\| \mathbf{1} - D^{-1}_{\boldsymbol{\widehat{\mu}}} \boldsymbol{\mu} \|_1  }{\| (I - \alpha A )^{-1} \|_1 n} \cdot \frac{\min_i \widehat\mu_i}{\max_i \widehat\mu_i},\label{eq:norm1}\\
&    \| \Delta\| \ge \alpha^{-1} \frac{\|\boldsymbol{\widehat{\mu}}-\boldsymbol{{\mu}}\|}{ \|(I - \alpha A )^{-1}\| \|\boldsymbol{\widehat{\mu}}\|},\label{eq:anynorm}
\end{align}
where $\|\cdot\|$ is any consistent norm.
\end{proposition}

\begin{proof}
By subtracting the two equalities we find that
\begin{equation}\label{eq:deltamuhat}
\alpha(I-\alpha A)^{-1}\Delta \boldsymbol{\widehat{\mu}}=\boldsymbol{\widehat{\mu}} - \boldsymbol{{\mu}},
\end{equation}
from which we obtain
\[
\alpha^{-1} \| \boldsymbol{\widehat{\mu}} - \boldsymbol{{\mu}}  \| = \| (I-\alpha A)^{-1} \Delta \boldsymbol{\widehat{\mu}} \| \le \| (I-\alpha A)^{-1} \| \| \Delta \| \cdot \| \boldsymbol{\widehat{\mu}} \|,
\]
that leads to \eqref{eq:anynorm}. By multiplying \eqref{eq:deltamuhat} by $D^{-1}_{\boldsymbol{\widehat{\mu}}}$, we obtain
\[
D^{-1}_{\boldsymbol{\widehat{\mu}}} (I-\alpha A)^{-1} \Delta D_{\boldsymbol{\widehat{\mu}}} \mathbf{1}= \alpha^{-1}(\mathbf{1}-
D^{-1}_{\boldsymbol{\widehat{\mu}}} {\boldsymbol{{\mu}}} ).
\]
Therefore
\[
\alpha^{-1} \| \mathbf{1}-
D^{-1}_{\boldsymbol{\widehat{\mu}}} {\boldsymbol{{\mu}}} \|_F \le \| D^{-1}_{\boldsymbol{\widehat{\mu}}} (I-\alpha A)^{-1} \Delta D_{\boldsymbol{\widehat{\mu}}} \|_F \cdot \| \mathbf{1} \|_F \le  \sqrt{n}\cdot \| \Delta \|_F \cdot \| (I-\alpha A)^{-1} \|_F \frac{\max_i \widehat\mu_i}{\min_i \widehat\mu_i},
\]
from which we obtain \eqref{eq:norm2}. Similarly, we proceed for \eqref{eq:norm1}.
\end{proof}

\subsection{Formulation as a Quadratic Programming (QP) problem}\label{sec:QP-Katz} Aiming at efficient solutions of problem \eqref{eq:Katz_1_alpha_1},  it is useful to reformulate such problem in the form of a Quadratic Programming (QP) problem; to this end, we use a generic pattern matrix $M \in \{0,1\}^{n \times n}$ such that $\Omega_M \subseteq \Omega_A$, {$M\mathbf{1}>\mathbf{0}$ (see Remark \ref{remark:generalpattern}).}  Problem~\eqref{eq:Katz_1_alpha_1} can be written as
\begin{equation}\label{eq:p1katz}
        \begin{array}{rl}
         \displaystyle \min_{\Delta \in \mathbb{S}(M) } & \displaystyle J(\Delta) = \beta \|\Delta\|_F^2 + (1-\beta) \|\Delta \|_1,  \\
        \text{s.t.} &  (I - \alpha(A + \Delta))\boldsymbol{\widehat{\mu}} - \mathbf{1} = \mathbf{0}, \\
        & -(A+\Delta) \leq 0.
    \end{array}  
\end{equation}
 As noted in \cite{gillis2024assigning}, a first reformulation of problem \eqref{eq:p1katz} can be written as follows:
\begin{equation}\label{eq:p1katz_ref}
        \begin{array}{rl}
         \displaystyle \min & \displaystyle J(\Delta) = \beta\|\operatorname{vec}(\Delta)\|_2^2+(1-\beta)\|\operatorname{vec}(\Delta)\|_1,  \\
        \text{s.t.} & (\boldsymbol{\widehat{\mu}}^{\top} \otimes I)\operatorname{vec}(\Delta) = \frac{1}{\alpha}(\boldsymbol{\widehat{\mu}}-\mathbf{1})-A\boldsymbol{\widehat{\mu}}, \\
        & \operatorname{diag}(\operatorname{vec}(\mathbf{1}\mathbf{1}^{\top} - M \circ \mathbf{1}\mathbf{1}^{\top}))\operatorname{vec}(\Delta)=\boldsymbol{0},   \\
        & -\operatorname{vec}(A) \leq \operatorname{vec}(\Delta).
    \end{array}  
\end{equation}
The reformulation in \eqref{eq:p1katz_ref} has $n^2$ variables, $n^2+n$ linear equalities and $n^2$ inequalities constraints.

We are looking for $\Delta $ having the sparsity pattern of $M$.
If the pattern we are considering has  $s_M\in O(n)$ {non-zeros}, it is computationally advantageous to reformulate problem \eqref{eq:p1katz_ref}  using only the variables associated to the possibly non-zero elements of $\Delta$. To this aim, let us define $\mathbf{x} \in \mathbb{R}^{s_M}$ as $\mathbf{x} = P_M \operatorname{vec}(\Delta) $ where $P_M \in \mathbb{R}^{s_M \times n^2}$ is the projector onto the pattern of $M$. Problem \eqref{eq:p1katz_ref} can hence be written as:
\begin{equation}\label{eq:p1katz_ref2}
        \begin{array}{rl}
         \displaystyle \min_{\mathbf{x}  \in \mathbb{R}^{s_M}} & \displaystyle  J({\mathbf{x}}) =\beta \|P^{\top}_{M}\mathbf{x}\|_2^2+(1-\beta)\|P^{\top}_{M}\mathbf{x}\|_1,  \\
        \text{s.t.} & (\boldsymbol{\widehat{\mu}}^{\top} \otimes I)P_M^{\top}\mathbf{x}  = \frac{1}{\alpha}(\boldsymbol{\widehat{\mu}}-\mathbf{1})-A\boldsymbol{\widehat{\mu}}, \\
        & -P_M \operatorname{vec}(A) \leq \mathbf{x}.
    \end{array}  
\end{equation}
Observe that for $(i,j)$ s.t. $[\Delta]_{i,j}=0$ and $ [A]_{i,j}\neq 0 $, i.e, in the case that $\Omega_{M} \subset \Omega_{A}$, the inequality $[\Delta]_{i,j} \geq -[A]_{i,j} $ is automatically satisfied  as $A \geq 0$, and, for this reason, can be eliminated from the inequalities constraints in~\eqref{eq:p1katz_ref}, leading to the inequality constraints in~\eqref{eq:p1katz_ref2}.  {The problem in \eqref{eq:p1katz_ref2} has $s_M\geq n$ variables, $n$ equality constraints and $s_M$ inequality constraints.} 

\begin{remark}
When $A=A^{\top}$ and the symmetry of the matrix $\Delta$ is added to the constraints in \eqref{eq:Katz_1_alpha_1}, it is enough to consider the vectorization of the lower triangular part of $\Delta$ and the relative projection onto the lower triangular part of the pattern.
\end{remark}

Since $P_M\in\{0,1\}^{s_M \times n^2 }$ and $P_MP_M^{\top}=I$, we have that $\|\mathbf{x}\|_2=\|P^{\top}_{M}\mathbf{x}\|_2$ and $\|\mathbf{x}\|_1=\|P^{\top}_{M}\mathbf{x}\|_1$. Hence, defining $\bar {\mathbf{x}} = \mathbf{x} + P_M \operatorname{vec}(A) $, we can write the problem in  \eqref{eq:p1katz_ref2} in the form

\begin{equation}\label{eq:QP1}
         {\mathcal{P}}^{\text{Katz}}_{\alpha,\beta} \, : \, \begin{array}{rl}
         \displaystyle \min_{\bar{\mathbf{x}}  \in \mathbb{R}^{s_M}} & \displaystyle J(\bar{\mathbf{x}}) =  \|\bar{\mathbf{x}}-P_M \operatorname{vec}(A))\|_2^2+ \tau  \| \bar{\mathbf{x}}- P_M \operatorname{vec}(A)  \|_1,  \\
        \text{s.t.} & (\boldsymbol{\widehat{\mu}}^{\top} \otimes I)P_M^{\top}\bar{\mathbf{x}}  = \frac{1}{\alpha}(\boldsymbol{\widehat{\mu}}-\mathbf{1})-A\boldsymbol{\widehat{\mu}}+ (\boldsymbol{\widehat{\mu}}^{\top} \otimes I)P_M^{\top} P_M \operatorname{vec}(A) , \\
        & \bar{\mathbf{x}} \geq \mathbf{0},
    \end{array}  
\end{equation}
where $\tau = (1-\beta)/\beta$. Employing a standard technique in sparse optimization, see e.g. \cite{MR4503349}, we introduce the nonnegative variables $\boldsymbol{\ell}^+=\max(\bar{\mathbf{x}}- P_M \operatorname{vec}(A),0)$ and $\boldsymbol{\ell}^-=\max(-(\bar{\mathbf{x}}- P_M \operatorname{vec}(A)),0)$; for these new variables the relations $\boldsymbol{\ell}^+ - \boldsymbol{\ell}^- = \bar{\mathbf{x}}- P_M \operatorname{vec}(A)$ and $ \| \bar{\mathbf{x}}- P_M \operatorname{vec}(A)  \|_1 = \mathbf{1}^{\top}\boldsymbol{\ell}^+ +\mathbf{1}^{\top}\boldsymbol{\ell}^-$ hold. Defining ${\mathbf{x}} = (\bar{\mathbf{x}};\boldsymbol{\ell}^+;\boldsymbol{\ell}^-) \in \mathbb{R}^{3s_M}$, we can write problem in \eqref{eq:QP1} in the standard QP form:
\begin{equation}\label{eq:p1katz_standard_form_L1_Final}
          \begin{array}{rl}
         \displaystyle \min_{{\mathbf{x}}  \in \mathbb{R}^{3s_M}} & \displaystyle  \frac{1}{2}{\mathbf{x}}^{\top}{Q}{\mathbf{x}}  + {\mathbf{c}}^{\top}{\mathbf{x}} \\
        \text{s.t.} & {L} {\mathbf{x}}   = {\mathbf{b}}, \\
        & {\mathbf{x}} \geq \mathbf{0},
    \end{array}  
\end{equation}
where 
\begin{equation}\label{eq:Katz_QP_Formulation_Blocks}
    \begin{split}
      {Q} & =\operatorname{blkdiag}(2I,0,0) { \; \in \mathbb{R}^{3s_M \times 3s_M } } , \quad 
      {\mathbf{c}} =(-2P_M \operatorname{vec}(A);\tau \mathbf{1};\tau \mathbf{1}) { \; \in \mathbb{R}^{3s_M }}, \\
      {L} & =[(\boldsymbol{\widehat{\mu}}^{\top} \otimes I)P_M^{\top},0,0; -I, I, -I] {\in \mathbb{R}^{(n+3s_M) \times 3s_M }} , \\ 
      {\mathbf{b}}& =\left(\frac{1}{\alpha}(\boldsymbol{\widehat{\mu}}-\mathbf{1})-A\boldsymbol{\widehat{\mu}}+ (\boldsymbol{\widehat{\mu}}^{\top} \otimes I)P_M^{\top} P_M \operatorname{vec}(A);-P_M \operatorname{vec}(A) \right) {\; \in \mathbb{R}^{n+3s_M }}.
    \end{split}
\end{equation}

{It is important to note that the matrix $Q$ in \eqref{eq:Katz_QP_Formulation_Blocks} is singular, and hence, that the solution of \eqref{eq:p1katz_standard_form_L1_Final} might not be unique.}

\section{Enforcing PageRank}
\label{sec:pagerank}
Our goal here is similar to that of Section~\ref{sec:katz} above, changing the centrality measure from Katz~\cite{Katz195339} to PageRank~\cite{Page1998107,MR3376760}. 
In this case, we may approach the problem with two strategies: either consider the stochastic matrix $G$ of \eqref{eq:G} and find a perturbation $\widetilde G$ of $G$ 
having the desired stationary vector 
$\widehat{\boldsymbol{\pi}}$, or look for a perturbation $\widetilde A$ of the adjacency matrix $A$ such that the corresponding PageRank stochastic matrix
has the desired stationary vector $\widehat{\boldsymbol{\pi}}$.
The first approach 
is similar to the one used in~\cite{gillis2024assigning} and, applied in this framework, completely ignores the pattern structure of $A$.
Here we follow the second approach since, from a modeling point of view, we believe that it is more appropriate to operate on the weights of the edges of the original network, rather than on the entries of the stochastic matrix~$G$.

Therefore, given $\alpha \in (0,1)$, the prescribed stationary vector $\widehat{\boldsymbol{\pi}}>\mathbf{0}$ such that $\widehat{\boldsymbol{\pi}}^{\top}\mathbf{1}=1$, the adjacency matrix  $A$ of the original graph $\mathcal G$, $D=\operatorname{diag}(A\mathbf{1})$, and $\mathbf{t}\ge \mathbf{0}$ such that $\mathbf{t}^{\top} \mathbf{1}=1$,
from \eqref{eq:pagerankvector2}, 
we address the problem of finding a perturbation $\Delta$ of  $A$
such that
\begin{equation} \label{eq:Pagerank_Original}
 \left(I - \alpha \left((D+\operatorname{diag}(\Delta \mathbf{1}))^{-1}(A+\Delta) \right)^{\top} \right) \widehat{\boldsymbol{\pi}} = (1 - \alpha) \mathbf{t},   
\end{equation}
under the constraint that $A+\Delta$ is nonnegative. 
Unlike the Katz centrality measure treated in Section~\ref{sec:katz}, the set on constrains is not linear in $\Delta$, hence we make the simplifying assumption that $\Delta\boldsymbol{1}=\boldsymbol{0}$, or, in other terms, that the degree of each node is unchanged. 

Finally, let us note that for a stochastic matrix with a fixed pattern, the conditions for which an assigned vector is the stationary vector are significantly more complex than those discussed in Proposition~\ref{pro:feasibility-conditions}; see in this regard~\cite{breen2015}.
 Therefore, by following~\cite{gillis2024assigning}, we look for a matrix $\Delta$ having zero entries corresponding to the zero entries of $A$, and  possibly entries different from zero on the diagonal, that is $\Delta\in\mathbb{S}(A+I)$.  
 Hence we allow the modification of the weights of existing edges and, possibly, the addition of some loops. {The addition of loops might not seem like a realistic graph modification. We point out that the choice $\Delta\in\mathbb{S}(A+I)$, in place of $\Delta\in\mathbb{S}(A)$, is used to prove the existence of the desired perturbation $\Delta$. In practice, there might be perturbations with small or null diagonal entries.}

For this purpose, {given $\alpha \in (0,1)$ and $\beta \in (0,1]$}, we first solve the optimization problem
\begin{equation*} \label{eq:Pagerank_complete_initial}
\mathcal{P}_{\alpha,\beta}^{\text{Pr}}: \; \begin{array}{rl}
\displaystyle \min_{\Delta \in \mathbb{S}(A + I) } & {J(\Delta) \; = \; }\beta \| \Delta \|_F^2 + (1-\beta) \|\operatorname{off-diag}(\Delta) \|_1   \\
\text{ s.t. } & \left(I - \alpha (\operatorname{diag}(A\mathbf{1})^{-1}(A+\Delta))^{\top} \right) \widehat{\boldsymbol{\pi}} = (1 - \alpha) \mathbf{t},\\
    & \Delta \mathbf{1} =\mathbf{0} \\
& \operatorname{off-diag}(A +  \Delta)  \geq 0,\\
\end{array}
\end{equation*}
then, from {a given optimal solution $\Delta^*$}, we build  $\widehat\alpha\in(0,1)$ and a stochastic $\widehat P\in\mathbb{S}(A+I)$ 
such that $\widehat{G}^{\top}\widehat{\boldsymbol{\pi}}=
\widehat{\boldsymbol{\pi}}$, where
$\widehat{G}=\widehat \alpha {\widehat P} + (1-\widehat\alpha) \mathbf{1}\mathbf{t}^{\top}$ (see Proposition~\ref{pro:interpretability}).

{As in the Katz's case, the} objective function $J(\Delta)$ in $\mathcal{P}^{\text{Pr}}_{\alpha,\beta}$ contains a trade-off between obtaining  small edge perturbations and obtaining a localized one given by the convex combination of the Frobenius norms and $\|\cdot\|_1$, respectively. The constraints require that the stationary vector is the desired one, that the perturbation does not modify the sum of the degrees of the nodes in the graph, and that only the off-diagonal weights of the edges of the modified network remain positive.  Under these conditions, we can now prove, exploiting an idea from~\cite{gillis2024assigning}, that the set of constraints is non-empty, i.e., that the problem is feasible, by demonstrating that a perturbation of a particular shape and that preserves the pattern is contained in the set.
\begin{proposition}\label{pro:feasibility-conditions-PR}
Given   $A\ge 0$ irreducible,  $\mathbf{t}\ge 0$ such that $\mathbf{t}^{\top}\mathbf{1}={1}$, $\alpha \in (0,1)$, 
$\boldsymbol{\widehat{\pi}}>\mathbf{0}$ such that $\boldsymbol{\widehat{\pi}}^{\top}\mathbf{1}=1$, then the set of matrices $\Delta\in\mathbb{S}(A+I)$  such that:
\begin{itemize}
    \item $\Delta \mathbf{1}=\mathbf{0}$,
    \item the off-diagonal entries of $A+\Delta$ are nonnegative,
    \item $\widetilde{G}^{\top} \boldsymbol{\widehat{\pi}}= \boldsymbol{\widehat{\pi}}$, where $\widetilde{G}=\alpha D^{-1}(A+\Delta)+(1-\alpha)\mathbf{1}\mathbf{t}^{\top}$ and $D=\operatorname{diag}(A\mathbf{1})$,
\end{itemize}
   is non-empty. 
\end{proposition}

\begin{proof}
We look for a matrix $\Delta$ of the kind $\Delta=D_{\boldsymbol{\sigma}}(A-D)$, where $D_{\boldsymbol{\sigma}}$ is the diagonal matrix with the vector $\boldsymbol{\sigma}=(\sigma_1,\ldots,\sigma_n)$ on the diagonal. The conditions $\Delta\in\mathbb{S}(A+I)$  and $\Delta \mathbf{1}=\mathbf{0}$ are obviously verified. The  off-diagonal entries of $A+\Delta$ are nonnegative if and only if $\boldsymbol{\sigma}\ge -\mathbf{1}$. We will show that there exists $\boldsymbol{\sigma}\ge -\mathbf{1}$ such that the condition
${\widetilde{G}}^\top \boldsymbol{\widehat{\pi}}= \boldsymbol{\widehat{\pi}}$ is satisfied.
The latter equation can be equivalently rewritten as 
\[
\left(
\alpha (D^{-1}A)^{\top}+\alpha \left( D^{-1}D_{\boldsymbol{\sigma}}(A-D) \right)^{\top} +(1-\alpha)\mathbf{t}\mathbf{1}^{\top} \right) \boldsymbol{\widehat{\pi}}=\boldsymbol{\widehat{\pi}},
\]
i.e.,
\[
\left( D^{-1}D_{\boldsymbol{\sigma}}(A-D) \right)^{\top}
\boldsymbol{\widehat{\pi}}
= \alpha^{-1}\left(  \left(I-\alpha D^{-1}A \right)^{\top} -(1-\alpha)\mathbf{t}\mathbf{1}^{\top} \right) \boldsymbol{\widehat{\pi}}.
\]
Since diagonal matrices commute and 
$\boldsymbol{\widehat{\pi}}^{\top} \mathbf{1}=1$, the latter equation reduces to
\begin{equation}\label{eq:sigmahatpi}
 (I-D^{-1}A)^{\top} D_{\boldsymbol{\widehat{\pi}}}
\boldsymbol{\sigma}
=\alpha^{-1}\left( (1-\alpha)\mathbf{t}-
(I-\alpha D^{-1}A)^{\top} \boldsymbol{\widehat{\pi}}\right).
\end{equation}
The matrix $I-D^{-1}A$ is a singular irreducible M-matrix~\cite[Definition~1.2, Chapter~6]{MR0544666} such that $(I-D^{-1}A)\mathbf{1}=\mathbf{0}$. Since $A$ is irreducible, we take can select $\mathbf{w}>\mathbf{0}$ as the only
$\mathbf{w}^{\top}(I-D^{-1}A)=0$ and $\mathbf{w}^{\top}\mathbf{1}=1$. The matrix $I-D^{-1}A+\mathbf{1}\mathbf{w}^{\top}$ is nonsingular and multiplying \eqref{eq:sigmahatpi} on the left by $(I-D^{-1}A+\mathbf{1}\mathbf{w}^{\top})^{-T}$ yields
\[
  (I-\mathbf{w}\mathbf{1}^{\top}) D_{\boldsymbol{\widehat{\pi}}} \boldsymbol{\sigma}=\alpha^{-1}
  (I-D^{-1}A+\mathbf{1}\mathbf{w}^{\top})^{-T}
\left(
(1-\alpha)\mathbf{t}-
(I-\alpha D^{-1}A)^{\top} \boldsymbol{\widehat{\pi}}
\right)
.
\]
From this expression, 
we find that a solution of \eqref{eq:sigmahatpi} is the vector
\[
\boldsymbol{\sigma}_*=\alpha^{-1}
D_{\boldsymbol{\widehat{\pi}}}^{-1}
(I-D^{-1}A+\mathbf{1}\mathbf{w}^{\top})^{-T}
\left(
(1-\alpha)\mathbf{t}-
(I-\alpha D^{-1}A)^{\top} \boldsymbol{\widehat{\pi}}
\right)
,
\]
while
all the solutions of \eqref{eq:sigmahatpi} are the vectors 
\[
\boldsymbol{\sigma}=
\boldsymbol{\sigma}_* +\gamma D_{\boldsymbol{\widehat{\pi}}}^{-1}\mathbf{w}
,
\]
where $\gamma$ is any real number. In particular, there is a value $\widehat\gamma$ such that, for any $\gamma>\widehat\gamma$, we have $\boldsymbol{\sigma}>-\mathbf{1}$, that concludes the proof of the existence of $\Delta$ satisfying the required properties. 
\end{proof}

\begin{remark}
    The assumptions on the pattern of $\Delta$ in Proposition~\ref{pro:feasibility-conditions-PR} can be slightly generalized, by asking that $\Delta\in\mathbb{S}(M+I)$, where $M$ is any nonnegative irreducible matrix, such that $M\mathbf{1}=A\mathbf{1}$. Indeed, the arguments of the proof still hold, by taking $\Delta=D_{\boldsymbol{\sigma}}(M-D)$.
\end{remark}

{It is important to note that any given solution of $\mathcal{P}_{\alpha,\beta}^{\text{Pr}}$} is not guaranteed to deliver a perturbed matrix $\widetilde{G}$ that is stochastic. Indeed, {despite} the off-diagonal entries of $\widetilde{G}$ {are guaranteed to be} nonnegative, we do not have information on the diagonal entries, as shown by the following example. 

\begin{example}\label{example:negative}
Consider the following case:
\[
A=\begin{bmatrix}
 0 & 1 & 0 \\
 1 & 0 & 1 \\
 0 & 1 & 0 \\
\end{bmatrix}, \;  \mathbf{t} = \nicefrac{1}{3}
\begin{bmatrix}
 1 \\
1 \\
 1 \\
\end{bmatrix}, \; \alpha=0.75.
\]
By using the notation of the proof of Proposition~\ref{pro:feasibility-conditions-PR}, we have
\[
D^{-1}A = \begin{bmatrix}
 0 & 1 & 0 \\
 \nicefrac{1}{2} & 0 & \nicefrac{1}{2} \\
 0 & 1 & 0 \\
\end{bmatrix}, \;\mathbf{w}=  \nicefrac{1}{4}
\begin{bmatrix}
   1 \\ 2 \\ 1
\end{bmatrix}, \;
(I-D^{-1}A+\mathbf{1}\mathbf{w}^{\top})^{-1}= \nicefrac{1}{8}
\begin{bmatrix}
7 & 2 & -1 \\
1 & 6 & 1\\
-1 & 2 & 7
\end{bmatrix}.
\]
By selecting $\boldsymbol{\widehat{\pi}}^{\top} =
\nicefrac{1}{3}\left[1,1,1\right]$, we obtain
$\boldsymbol{{\sigma}}_*^{\top}=
\nicefrac{1}{6}\left[5,2,5\right]$. 
By choosing $\boldsymbol\sigma={\boldsymbol{\sigma}}_*$, we find that neither
$A+\Delta$, with $\Delta=D_{\boldsymbol{\sigma}}(A-D)$, nor $\widetilde G$ are nonnegative, indeed,
\[
A+D_{\boldsymbol{\sigma}}(A-D)=\nicefrac{1}{6}
\begin{bmatrix}
    -5 & 11 & 0 \\
    8 & -4 & 8\\
    0 & 11 & -5
\end{bmatrix}, \;
\widetilde G=\nicefrac{1}{24}
\begin{bmatrix}
    -13 & 35 & 2 \\
    14 & -4 & 14 \\
    2 & 35 & -13
\end{bmatrix}.
\]
Moreover, $\widetilde G$ has eigenvalues $1, -\nicefrac{5}{8}, -\nicefrac{13}{8}$, therefore the spectral radius of $\widetilde G$ is greater than~$1$.
We may try to look for a different value of $\boldsymbol{\sigma}$ such that
the entries of $A+D_{\boldsymbol{\sigma}}(A-D)$ are nonnegative.
From the proof of Proposition~\ref{pro:feasibility-conditions-PR}, 
the set of vectors $\boldsymbol{\sigma}$ such that $\widetilde G^{\top} \widehat{\boldsymbol{\pi}}=
\widehat{\boldsymbol{\pi}}$ and the off-diagonal entries of $A+D_{\boldsymbol{\sigma}}(A-D)$ are nonnegative is given by $\boldsymbol\sigma={\boldsymbol{\sigma}}_*+\gamma D_{\boldsymbol{\widehat{\pi}}}^{-1}\mathbf{w}$, where $\gamma$ is such that $\boldsymbol\sigma\ge -\mathbf{1}$. Therefore $\gamma$ must satisfy  the condition $\gamma \ge \max_i\left\{(-1-[\boldsymbol{\sigma}_*]_i)\cdot [\widehat{\boldsymbol{\pi}}]_i/
[\mathbf{w}]_i\right\}=-\nicefrac{8}{9}$. On the other hand, if we impose the condition that the diagonal entries of 
$A+D_{\boldsymbol{\sigma}}(A-D)$ are nonnegative as well, we find that $\gamma$ must satisfy $\gamma\le -\nicefrac{10}{9}$. Therefore, for this example, all the vectors $\boldsymbol{\sigma}$ which guarantee that the off-diagonal entries of $A+D_{\boldsymbol{\sigma}}(A-D)$ are nonnegative, are such that at least one diagonal entry of this matrix is negative. Nevertheless, we will show in Proposition~\ref{pro:interpretability} that it is possible to enforce the nonnegativity also of the diagonal entries  by modifying the parameter $\alpha$ defining the PageRank problem, and by shifting-and-scaling the matrix~$\widetilde G$.
\end{example}
{Given a matrix $\Delta$ that satisfies the conditions of Proposition~\ref{pro:feasibility-conditions-PR} but for which some diagonal elements of $D^{-1}(A+\Delta)$ are negative, we can apply the following proposition to shift and scale the matrix $\widetilde G$. This transformation restores the required nonnegativity properties while preserving the desired stationary vector $\widehat{\boldsymbol{\pi}}$.}
\begin{proposition}\label{pro:interpretability}
Given a matrix $\Delta$ satisfying the conditions of Proposition~\ref{pro:feasibility-conditions-PR}, define
$\theta=\min_i ([D^{-1}(A+\Delta)]_{i,i})$.
If
$\theta \ge 0$, then $D^{-1}(A+\Delta)$ is stochastic.
Otherwise, if $\theta<0$,
by setting $\widehat r=1-\alpha \theta$, then
for any $r\ge \widehat r$ we have $\widehat{G}^{\top}\boldsymbol{\widehat{\pi}}  =\boldsymbol{\widehat{\pi}}$, where
  \begin{equation}\label{eq:hatPPR}
\widehat{G}=\widehat \alpha \widehat P + (1-\widehat\alpha) \mathbf{1}\mathbf{t}^{\top},
\end{equation}
and
\[
\widehat\alpha=1-\frac{1-\alpha}{r},~~
\widehat P=\frac{1}{r-1+\alpha}
\left(\alpha D^{-1}(A+\Delta)+(r-1) I\right),
\]
with $\widehat P$ stochastic.
\end{proposition}

\begin{proof}
If $\theta\ge 0$ then all the entries of $D^{-1}(A+\Delta)$ are nonnegative; since 
$(D^{-1}(A+\Delta))\mathbf{1}=\mathbf{1}$, the matrix is stochastic.
If $\theta<0$, let $r>0$ and define
\[
\widehat{G}=\frac{1}{r}(\widetilde{G}-I)+I.
\]
Since $(I-\widetilde{G})\mathbf{1}=\mathbf{0}$ and $(I-\widetilde{G})^{\top} \boldsymbol{\widehat{\pi}}=\mathbf{0}$, then
$\widehat{G}\mathbf{1}=\mathbf{1}$ and $\widehat{G}^{\top} \boldsymbol{\widehat{\pi}} =\boldsymbol{\widehat{\pi}}$.
We look for values of $r$ such that $\widehat{G}$ can be written as in \eqref{eq:hatPPR}, where $\widehat P$ is stochastic.
By replacing the expression of $\widetilde{G}$ in $\widehat{G}$, we find that
\[
\widehat{G}=\frac{1}{r}\left(
\alpha D^{-1}(A+\Delta)-I
\right)+I+\frac{1-\alpha}{r}\mathbf{1}\mathbf{t}^{\top}.
\]
The off-diagonal entries of $\alpha D^{-1}(A+\Delta)-I$ are nonnegative,  and 
$(\alpha D^{-1}(A+\Delta)-I)\mathbf{1}=(\alpha-1)\mathbf{1}$, thus the diagonal entries of $\alpha D^{-1}(A+\Delta)-I$ are smaller than $\alpha-1<0$.
Therefore 
if $r\ge\widehat r$, where
$\widehat r=1- \alpha \theta$,
then
the matrix 
\[
F=\frac{1}{r}\left(
\alpha D^{-1}(A+\Delta)-I
\right)+I
\]
is nonnegative, since the off-diagonal entries are nonnegative and the diagonal entries of $\frac1r \left(
\alpha D^{-1}(A+\Delta)-I
\right)$ belong to the interval $[-1,0]$ by construction.
By setting $1-\widehat\alpha=
\frac{1-\alpha}{r}$ and $\widehat P=\widehat\alpha^{-1}F$, we arrive the expression for $\widehat P$, which is stochastic if $r\ge \widehat r$.
\end{proof}

\begin{remark}\label{rema:r_hat}
Proposition \ref{pro:interpretability} shows that, {if $\Delta^*$ solving $\mathcal{P}_{\alpha,\beta}^{\text{Pr}}$ is such that any of the diagonal entries of $A+\Delta^*$ is negative,} i.e., $\theta<0$,  the perturbation resulting from the optimization procedure can be {interpreted} as a modification of the original random walk on the graph with a modified teleportation parameter $\widehat{\alpha}$ but corresponding to the same preference vector $\mathbf{t}$. Moreover, when $\widehat{r}$ is close to $1$, the value $\widehat{\alpha}$ in \eqref{eq:hatPPR} is close to the original $\alpha$. 
In the case where $a_{i,i}=0$ for $i=1,\ldots,n$, we have 
\[
|\theta|=-\min_i \frac{[{\Delta^*}]_{i,i}}{d_i} =
\max_i\left( \frac{\sum_{j\ne i}[{\Delta^*}]_{i,j}}{d_i} \right)\le \|\operatorname{off-diag}({\Delta^*}) \|_\infty \frac{1}{\min_i d_i}.
\]
Therefore, if ${\Delta^*}$ has small norm with respect to ${\min_i d_i}$, then $\theta$ is close to zero, hence $\hat r=1-\alpha\theta$ is close to 1, and $\hat\alpha$ is close to $\alpha$.
The numerical results presented in Section~\ref{sec:enforcing_pagerank} confirm that on real-world problems $\widehat{r} \approx 1$. 

\end{remark}

We can formulate a bound analogous to the one in Proposition~\ref{prop:bound_delta} for the Katz centrality measure and this perturbation.
\begin{proposition}
Assume that $\boldsymbol{\pi}>\mathbf{0}$ satisfies $ \left( \alpha (D^{-1}A)^{\top} + (1 - \alpha) \mathbf{t}\mathbf{1}^{\top} \right) \boldsymbol{\pi} =\boldsymbol{\pi}  $, and that $\widehat{\boldsymbol{\pi}}>\mathbf{0}$ satisfies $ \left(\alpha (D^{-1}(A + \Delta))^{\top} + (1 - \alpha) \mathbf{t} \mathbf{1}^{\top}\right) \widehat{\boldsymbol{\pi}} = \widehat{\boldsymbol{\pi}}$, {with $\boldsymbol{\pi}^{\top} \mathbf{1}=
\widehat{\boldsymbol{\pi}}^{\top} \mathbf{1}=1$}, then
\[
\| \Delta \|_1 \geq \alpha^{-1}\frac{\| \widehat{\boldsymbol{\pi}} - {\boldsymbol{\pi}}\|_1}{ \| (I-\alpha D^{-1}A)^{-1}\|_\infty}\frac{1}{\max_i\{\widehat{\boldsymbol{\pi}}_i/d_i\}}\geq \frac{1-\alpha}{\alpha } \frac{\| \widehat{\boldsymbol{\pi}} - {\boldsymbol{\pi}}\|_\infty}{\max_i\{\widehat{\boldsymbol{\pi}}_i/d_i\}}.
\]
\end{proposition}

\begin{proof}
Subtract the two identities for $\boldsymbol{\pi}$ and $\widehat{\boldsymbol{\pi}}$ employing that $\boldsymbol{\pi}^{\top} \mathbf{1}=  \widehat{\boldsymbol{\pi}}^{\top} \mathbf{1}=1$
\begin{eqnarray*}
 (I - \alpha D^{-1}A)^{\top} ( \boldsymbol{\pi} - \widehat{\boldsymbol{\pi}} ) & = & \alpha  \Delta^{\top} D^{-1} \widehat{\boldsymbol{\pi}},\\
( \boldsymbol{\pi} - \widehat{\boldsymbol{\pi}} ) & = & \alpha  (I - \alpha D^{-1}A)^{-T}\Delta^{\top} D^{-1} \widehat{\boldsymbol{\pi}},
\end{eqnarray*}
from which we find--- employing the submultiplicativity of the norm,
\[
\| \Delta \|_1 \geq \alpha^{-1}\frac{\| \widehat{\boldsymbol{\pi}} - {\boldsymbol{\pi}}\|_1}{ \| (I-\alpha D^{-1}A)^{-1}\|_\infty}\frac{1}{\max_i\{\widehat{\boldsymbol{\pi}}_i/d_i\}}
\]
and the latter follows from $\|(I-Q)^{-1}\|_\infty \leq (1 - \|Q\|_\infty)^{-1}$ whenever $\|Q\|_\infty < 1$.
\end{proof}

\subsection{Formulation as a QP problem}\label{sec:QP-PageRank}
{In the sequel, we assume $\alpha>0$, $\alpha \rho(A)<1$ and $0<\beta\le 1$ fixed.} It is important to note that given the presence of the constraint $\Delta \mathbf{1}=\mathbf{0}$, problem $\mathcal{P}_{\alpha,\beta}^{\text{Pr}}$ in the previous section can be reformulated as
\begin{equation} \label{eq:Pagerank_complete}
\mathcal{P}_{\alpha,\beta}^{\text{Pr}}: \; \begin{array}{rl}
\displaystyle \min_{\Delta \in \mathbb{S}(A + I) } & {J(\Delta) \; =  \; } \beta \| \Delta \|_F^2 + (1-\beta) \|\operatorname{off-diag}(\Delta)  \|_1   \\
\text{ s.t. } & \left(I - \alpha (\operatorname{diag}(A\mathbf{1})^{-1}(A+\Delta))^{\top} \right) \widehat{\boldsymbol{\pi}} = (1 - \alpha) \mathbf{t},\\
    & \Delta \mathbf{1} =\mathbf{0} \\
& \hbox{off-diag}(A +  \Delta)  \geq 0.\\
\end{array}
\end{equation}
As done is Section~\ref{sec:QP-Katz} we reformulate $\mathcal{P}_{\alpha,\beta}^{\text{Pr}}$ for a generic pattern $\Omega_{M + I}$ such that $\Omega_{M + I} \subseteq \Omega_{A+I}$ where $M \in \{0,1\}^{n \times n}$.
We have then,
\[
\mathcal{P}_{\alpha,\beta}^{\text{Pr}}: \; \begin{array}{rl}
\displaystyle \min_{\Delta \in \mathbb{S}(M + I) } & J(\Delta) =  \beta \| \Delta \|_F^2 + (1-\beta)  \| \operatorname{off-diag}(\Delta)  \|_1   \\
\text{ s.t. } &  \Delta^{\top} \operatorname{diag}(A\mathbf{1})^{-1}\widehat{\boldsymbol{\pi}}= \frac{1}{\alpha}(\widehat{\boldsymbol{\pi}}-(1-\alpha)\mathbf{t})-A^{\top}\operatorname{diag}(A\mathbf{1})^{-1}\widehat{\boldsymbol{\pi}}, \\
    & \Delta \mathbf{1} =\mathbf{0} \\
& \hbox{off-diag}(A +  \Delta)  \geq 0.\\
\end{array}.
\]
Considering again the vectorization of the matrix $\Delta$, we have:
\begin{equation}\label{eq:pagerank_ref}
        \begin{array}{rl}
         \displaystyle \min & \displaystyle J(\Delta) = \beta\|\operatorname{vec}(\Delta)\|_2^2+(1-\beta)\|\operatorname{diag}(\operatorname{vec}(\mathbf{1}\mathbf{1}^{\top} -I \circ \mathbf{1}\mathbf{1}^{\top}))\operatorname{vec}(\Delta)\|_1,    \\
        \text{s.t.} & ((\operatorname{diag}(A\mathbf{1})^{-1}\widehat{\boldsymbol{\pi}})^{\top} \otimes I)\operatorname{vec}(\Delta^{\top}) = \frac{1}{\alpha}(\widehat{\boldsymbol{\pi}}-(1-\alpha)\mathbf{t})-A^{\top}\operatorname{diag}(A\mathbf{1})^{-1}\widehat{\boldsymbol{\pi}}, \\
        &  (\mathbf{1}^{\top} \otimes I)\operatorname{vec}(\Delta)=\mathbf{0},\\
        & \operatorname{diag}(\operatorname{vec}(\mathbf{1}\mathbf{1}^{\top} - (M+I) \circ \mathbf{1}\mathbf{1}^{\top}))\operatorname{vec}(\Delta)=\boldsymbol{0},   \\
        & - \operatorname{diag}(\operatorname{vec}(\mathbf{1}\mathbf{1}^{\top} - I \circ \mathbf{1}\mathbf{1}^{\top}))\operatorname{vec}(A) \leq \operatorname{diag}(\operatorname{vec}(\mathbf{1}\mathbf{1}^{\top} -I \circ \mathbf{1}\mathbf{1}^{\top}))\operatorname{vec}(\Delta).
    \end{array}  
\end{equation}
{As we have done in Katz's case, if the pattern we are considering has  $s_{M+I} \in O(n)$ non-zeros, it is convenient to define} $\mathbf{x} \in \mathbb{R}^{s_{M+I}}$ as $\mathbf{x} = P_{M+I} \operatorname{vec}(\Delta) $ where $P_{M+I} \in \mathbb{R}^{s_{M+I} \times n^2}$ is the projector onto the pattern of $\Delta$ and $K$ as the orthogonal commutation matrix  such that $K\operatorname{vec}(\Delta)=\operatorname{vec}(\Delta^{\top})$. Let us define, moreover, $k_i:=i+(i-1)n$ for $i=1,\dots,n$, i.e., $\operatorname{vec}(\Delta)_{k_i}=\Delta_{ii}$ and $\bar{k}_i$ as the indices satisfying $$\mathbf{x}_{\bar{k}_i} = (P_{M+I} \operatorname{vec}(\Delta))_{\bar{k}_i} = \Delta_{ii}.$$
Finally, let us define $\mathcal{F}:= \{ \bar{k}_1, \dots, \bar{k}_n \}$ and $\mathcal{C}:=\{1, \dots,  s_{M+I}\} \setminus \mathcal{F}$. Observing that $$\|\operatorname{diag}(\operatorname{vec}(\mathbf{1}\mathbf{1}^{\top} -I \circ \mathbf{1}\mathbf{1}^{\top}))\operatorname{vec}(\Delta)\|_1=\|(P_{M+I}\operatorname{vec}(\Delta))_{\mathcal{C}}\|_1,$$ we can then rewrite  Problem~\eqref{eq:pagerank_ref}~as:
\begin{equation}\label{eq:pagerank_vec}
        \begin{array}{rl}
         \displaystyle \min_{\mathbf{x}  \in \mathbb{R}^{s_{M+I}}} & \displaystyle \beta \| \mathbf{x}\|_2^2+(1-\beta)\|\mathbf{x}_{\mathcal{C}}\|_1,  \\
        \text{s.t.} & ((\operatorname{diag}(A\mathbf{1})^{-1}\widehat{\boldsymbol{\pi}})^{\top} \otimes I)KP_{M+I}^{\top}\mathbf{x} = \frac{1}{\alpha}(\widehat{\boldsymbol{\pi}}-(1-\alpha)\mathbf{t})-A^{\top}\operatorname{diag}(A\mathbf{1})^{-1}\widehat{\boldsymbol{\pi}}, \\
        &  (\mathbf{1}^{\top} \otimes I)P_{M+I}^{\top}\mathbf{x} =\mathbf{0},\\
       & -(P_{M+I}\operatorname{vec}(A))_i \leq \mathbf{x}_i \hbox{ if } i \in \mathcal{C}, \\
       &  \mathbf{x}_i \hbox{ free if } i \in \mathcal{F}.
    \end{array}  
\end{equation}
Also in this case, defining $\bar {\mathbf{x}} = \mathbf{x} + \mathbf{a} $ where $\mathbf{a}:=P_{M+I}\operatorname{diag}(\operatorname{vec}(\mathbf{1}\mathbf{1}^{\top} - I \circ \mathbf{1}\mathbf{1}^{\top}))\operatorname{vec}(A)$, and observing that $\mathbf{a}_i=0$ if $i \in \mathcal{F}$, we can write the problem in  \eqref{eq:pagerank_vec} in the standard form
\begingroup
\allowdisplaybreaks
\begin{equation}\label{eq:QP1_PR}
         {\mathcal{P}}^{\text{Pr}}_{\alpha,\beta}  : \hspace{-1em} \begin{array}{rl}
         \displaystyle \min_{\bar{\mathbf{x}}  \in \mathbb{R}^{s_{M+I}}} & \displaystyle J(\bar{\mathbf{x}}) =  \| \bar{\mathbf{x}}-\mathbf{a}\|_2^2 +\tau \|\bar{\mathbf{x}}_{\mathcal{C}}-\mathbf{a}_{\mathcal{C}}\|_1,  \\
        \text{s.t.}  & ((\operatorname{diag}(A\mathbf{1})^{-1}\widehat{\boldsymbol{\pi}})^{\top} \otimes I)KP_{M+I}^{\top}\bar{\mathbf{x}} = \\ &\frac{1}{\alpha}(\widehat{\boldsymbol{\pi}}-(1-\alpha)\mathbf{t})-A^{\top}\operatorname{diag}(A\mathbf{1})^{-1}\widehat{\boldsymbol{\pi}}+((\operatorname{diag}(A\mathbf{1})^{-1}\widehat{\boldsymbol{\pi}})^{\top} \otimes I)KP_{M+I}^{\top}\mathbf{a}, \\
        &  (\mathbf{1}^{\top} \otimes I)P_{M+I}^{\top}\bar{\mathbf{x}} =(\mathbf{1}^{\top} \otimes I)P_{M+I}^{\top}\mathbf{a},\\
        & \bar{\mathbf{x}}_i \geq 0 \hbox{ if } i \in \mathcal{C}, \\
        & \bar{\mathbf{x}}_i \hbox{ free if } i \in \mathcal{F}. 
    \end{array}  
\end{equation}
\endgroup
where $\tau = (1-\beta)/\beta$. As previously done in the Katz case, we introduce the nonnegative variables $\boldsymbol{\ell}^+=\max(\mathbf{x}_{\mathcal{C}}- \mathbf{a}_{\mathcal{C}},0)$ and $\boldsymbol{\ell}^-=\max(-(\mathbf{x}_{\mathcal{C}}- \mathbf{a}_{\mathcal{C}}),0)$. We have, again, $\boldsymbol{\ell}^+ - \boldsymbol{\ell}^- = \mathbf{x}_{\mathcal{C}}- \mathbf{a}_{\mathcal{C}}$ and $ \| \mathbf{x}_{\mathcal{C}}- \mathbf{a}_{\mathcal{C}}  \|_1 = \mathbf{1}^{\top}\boldsymbol{\ell}^+ +\mathbf{1}^{\top}\boldsymbol{\ell}^-$. {Defining ${\mathbf{x}} = (\bar{\mathbf{x}};\boldsymbol{\ell}^+;\boldsymbol{\ell}^-) \in  \mathbb{R}^{s_{M+I} + 2 |\mathcal{C}|}$, with $|\mathcal{C}| = s_{M+I}-n$,} we can write problem in~\eqref{eq:QP1_PR} in the standard QP form as
\begin{equation}\label{eq:p1PR_standard_form_Final}
          \begin{array}{rl}
         \displaystyle \min_{ {{\mathbf{x}}  \in\mathbb{R}^{s_{M+I} + 2 |\mathcal{C}|}}} & \displaystyle  \frac{1}{2}{\mathbf{x}}^{\top}{Q}{\mathbf{x}}  + {\mathbf{c}}^{\top}{\mathbf{x}} \\
        \text{s.t.} & {L} {\mathbf{x}}   = {\mathbf{b}}, \\
         & {\mathbf{x}}_i\geq 0, \hbox{ if } i \in \widehat{\mathcal{C}}, \\
        & {\mathbf{x}}_i \hbox{ free if } i \in \widehat{\mathcal{F}}
    \end{array}  
\end{equation}
where
\begin{equation} \label{eq:PR_QP_formulation_Blocks}
\begin{split}
& \widehat{\mathcal{F}} =\mathcal{F}, \quad
\widehat{\mathcal{C}} = {\mathcal{C} \bigcup \{s_{M+I}+1,s_{M+I}+2,\dots,  3s_{M+I}-2n\}} \\ 
&  {Q} =\operatorname{blkdiag}(2I,0,0) \in {\mathbb{R}^{(s_{M+I} + 2 |\mathcal{C}|) \times (s_{M+I} + 2 |\mathcal{C}|)}},\\
& {\mathbf{c}} =(-2\mathbf{a};\tau \mathbf{1}_{|\mathcal{C}|};\tau \mathbf{1}_{|\mathcal{C}|}) \in \mathbb{R}^{s_{M+I} + 2 |\mathcal{C}|} , \\
 & {L}  =  \begin{bmatrix}
 (\mathbf{r}^{\top} \otimes I)KP_{M+I}^{\top} & 0 & 0 \\
    (\mathbf{1}^{\top} \otimes I)P_{M+I}^{\top} & 0 & 0 \\ 
    -I_{|\mathcal{C}|\times s_{M+I} } & I_{|\mathcal{C}|} & -I_{|\mathcal{C}|}
 \end{bmatrix} {\in \mathbb{R}^{(s_{M+I} + 2|\mathcal{C}|) \times (2n+|\mathcal{C}|)}}, \\
    & \mathbf{b} =\begin{bmatrix}
    \frac{1}{\alpha}(\widehat{\boldsymbol{\pi}}-(1-\alpha)\mathbf{t})-A^{\top}\mathbf{r}+(\mathbf{r}^{\top} \otimes I)KP_{M+I}^{\top}\mathbf{a} \\
    (\mathbf{1}^{\top} \otimes I)P_{M+I}^{\top}\mathbf{a} \\
    - \mathbf{a}_\mathcal{C}
    \end{bmatrix} {\in \mathbb{R}^{ 2n+|\mathcal{C}|}} ,\\
& {\mathbf{r} =  \operatorname{diag}(A\mathbf{1})^{-1}\widehat{\boldsymbol{\pi}},}\\   
\end{split}
\end{equation}
and where  $(I_{|\mathcal{C}|\times s_{M+I}})_{ij}=1 $ if $i,j \in \mathcal{C}$ and zero otherwise. {Also in this case, it is important to note that the matrix $Q$ in \eqref{eq:PR_QP_formulation_Blocks} is singular, and hence, that the solution of \eqref{eq:p1PR_standard_form_Final} might not be unique.}

\section{Numerical examples}
\label{sec:numerics}

The code to generate the examples and algorithms discussed in the paper is available on the \texttt{GitHub} repository  \href{https://github.com/Cirdans-Home/enforce-katz-and-pagerank}{Cirdans-Home/enforce-katz-and-pagerank}. The numerical examples are executed on a \texttt{lnx2} node of the Toeplitz cluster at the Green Data Center of the University of Pisa, this is equipped with an Intel\textsuperscript{\textregistered} Xeon\textsuperscript{\textregistered} CPU E5-2650 v4 at \qty{2.20}{\giga\hertz} with 2 threads per core, 12 cores per socket and 2 socket, and \qty{256}{\giga\byte} of RAM.

\subsection{Target Distributions}\label{sec:target-distribution}
Our analysis will cover both, the algorithmic details connected to the performance of the chosen solver, and the graph-related characteristics of the obtained solutions. Aiming at showcasing the robustness of our algorithmic approach with respect to the target distribution $\widehat{\boldsymbol{\mu}}/\widehat{\boldsymbol{\pi}}$, we consider two possible scenarios 
for the choices of  $\widehat{\boldsymbol{\mu}}/\widehat{\boldsymbol{\pi}}$:
\begin{description}
    \item[\textbf{S1}:] $\widehat{\boldsymbol{\mu}}/\widehat{\boldsymbol{\pi}}$ puts the top $10 \%$ of the nodes to their averaged value in ${\boldsymbol{\mu}}/{\boldsymbol{\pi}}$;
    \item[\textbf{S2}:] $\widehat{\boldsymbol{\mu}}/\widehat{\boldsymbol{\pi}}$ reverts the rank of the top $10\%$ of nodes in ${\boldsymbol{\mu}}/{\boldsymbol{\pi}}$.
\end{description}

\subsection{Terms of comparison}
When considering the vector of scores obtained using the perturbation $\Delta$, in principle, we could measure the norm of the difference between the computed ones and the desired ones---respectively $\widehat{\boldsymbol{\mu}}$ and $\widehat{\boldsymbol{\pi}}$. However, what is really of interest about the ranking obtained following the computed perturbation is the ordering of the vector entries rather than their exact values.

To this end, we consider the Kendall correlation coefficient~\cite{Kendall} to measure the similarity between the target distributions and the computed one after introducing the perturbation $\Delta$. This is a measure of association between any two rankings of the same sets and is especially tailored for ordinal data, where the order of the items matters but not necessarily the exact values. Given two rankings $\boldsymbol{\mu}$ and $\widehat{\boldsymbol{\mu}}$ of length $n$, with ties allowed, Kendall's $\kappa_\tau$
\begin{equation}\label{eq:kendall-tau}
    \kappa_\tau = \frac{{\text{``number of concordant pairs''} - \text{``number of discordant pairs''}}}{{\binom{n}{2}}} \,\in\,[-1,1]
\end{equation}
where a pair of elements $(i, j)$ is concordant if their order is preserved in both rankings, and discordant if their order is reversed. Ties are counted as half a concordant or discordant pair. {See Section \ref{sec:kendal} for more details about the Kendall's  correlation coefficient.}

\subsection{Dataset}

In Table \ref{tab:dataset} we present the details of the dataset considered in our numerical experiments. The networks are taken from the \texttt{SuiteSparse} collection~\cite{SuiteSparseCollection} and represent graphs coming from different applications, and with different structures. To avoid reporting the network name in all figures and tables, we have numbered the test cases sequentially and will use the assigned order number to identify them.
\begin{table}[htbp]
\caption{Dataset Details. The matrices are from the \texttt{SuiteSparse} collection~\cite{SuiteSparseCollection}. The table reports the name, the graph type {(``comb. prob.'' for combinatorial problem, ``undir./dir.'' for undirected/directed, ``wtd'' for weighted, ``multigr.'' for multigraph)} , the number of nodes, and the number of non-zeroes entries. All the selected graphs have one connected/strongly connected component.}\label{tab:dataset}
\centering\small
\begingroup
\setlength{\tabcolsep}{1.2pt}
\begin{tabular}{rlllr|rlllr}
\toprule
& Name & Type & $n$ & $\operatorname{nnz}$ & & Name & Type & $n$ & $\operatorname{nnz}$ \\
\midrule
 1 &    EX5      &  comb. prob.     &   6545   & 295680    & 9 &    de2010      &  undir. wtd graph     &   24115   & 116056    \\   
 2 &    PGPgiantcompo      &  undir. multigr.     &   10680   & 48632    & 10 &    delaunay\_n16      &  undir. graph     &   65536   & 393150    \\   
 3 &    cage10      &  dir. wtd graph     &   11397   & 150645    & 11 &    fe\_4elt2      &  undir. graph     &   11143   & 65636    \\   
 4 &    cage11      &  dir. wtd graph     &   39082   & 559722   & 12 &    gre\_1107      &  dir. wtd graph     &   1107   & 5664    \\   
 5 &    cs4      &  undir. graph     &   22499   & 87716    & 13 &    nh2010      &  undir. wtd graph     &   48837   & 234550    \\   
 6 &    ct2010      &  undir. wtd graph     &   67578   & 336352    & 14 &    uk      &  undir. graph     &   4824   & 13674    \\   
 7 &    cti      &  undir. graph     &   16840   & 96464    & 15 &    vt2010      &  undir. wtd graph     &   32580   & 155598    \\   
 8 &    data      &  undir. graph     &   2851   & 30186    &  &  &  &  &  \\ 
\bottomrule
\end{tabular}
\endgroup
\end{table}

\subsection{Solution of the QPs}\label{sec:ipm-solve}
In our numerical experiments, we solve the QP problems presented in Sections~\ref{sec:QP-Katz} and~\ref{sec:QP-PageRank} using  PS-IPM (Proximal Stabilised-Interior Point Method), see \cite{MR4594481,CIPOLLA2023} for related theoretical details. This method is particularly well-suited for problems characterised by inherent ill-conditioning of the problem's data. The presence of the Proximal-Stabilization induces a Primal-Dual Regularization that is particularly helpful for the numerical back-solvers used for the IPM-related linear systems. Indeed, the computationally intensive part of IPMs for the QPs here considered is represented by the solution of a linear system of the form
\begin{equation} \label{eq:schur_IPM}
S_{\rho, \delta}={L}({Q} + \Theta^{-1} +\rho I)^{-1}{L}^{\top} +\delta I     
\end{equation}
where $\Theta^{-1}={X}^{-1}{S} $ is a diagonal IPM iteration dependent matrix responsible for the identification of the \textit{active-variables}. Such a matrix has the unpleasant feature of exhibiting some of the elements converging to zero and some of the elements converging to infinity when the optimization process approaches convergence to the optimum. When the Primal-Dual Regularization is not considered, i.e., when $\rho=\delta=0$, it is possible to prove that the matrices $L\Theta L^{\top}$ have diverging condition numbers. As a consequence, the related Newton directions could be affected by numerical instabilities responsible for preventing the successful numerical convergence of the method. Instead, as it is apparent from \eqref{eq:schur_IPM}, the use of Primal-Dual regularization, i.e., when $\rho>0$ and $\delta>0$, allows a uniform control on the conditioning of the matrices $S_{\rho,\delta}$. As a consequence, the related IPMs are able to solve successfully also problems characterized by an inherent ill-conditioning, such as those arising when enforcing the PageRank centrality vector, see, e.g., the first column of Tables~\ref{tab:S1-PR_tab1},~\ref{tab:S1-PR_tab2},~\ref{tab:S2-PR_tab1} and~\ref{tab:S2-PR_tab2}. It is important to note, moreover, that in our numerical experiments, we solve the linear systems involving matrices as in \eqref{eq:schur_IPM} using \texttt{Matlab}'s \texttt{chol} function. As the numerical results presented in the next sections will show, for sufficiently sparse graphs, the Cholesky factors do not show significant fill-in, see the column 
``\texttt{nnz Chol.} of eq. \eqref{eq:schur_IPM}'' in all the tables presented in the remainder of this work. %
{It is crucial to recognize that solving the linear system \eqref{eq:schur_IPM}—and in particular, computing its associated Cholesky factorization—dominates the overall computational cost of our approach. This complexity serves as a key indicator of our method's scalability. In general, the dimension of \eqref{eq:schur_IPM} corresponds to the number of linear constraints in the quadratic program in standard form, as reflected by the first dimension of the matrix $L$ in \eqref{eq:p1katz_standard_form_L1_Final} and \eqref{eq:p1PR_standard_form_Final}. As discussed in Sections \ref{sec:QP-Katz} and \ref{sec:QP-PageRank}, assuming that the sparsity pattern of the perturbation contains $O(n)$ non-zero elements, the number of such constraints is also $O(n)$ (with the constant factor depending on whether we consider the Katz or PageRank model and on the impact of the $1$-norm induced sparsification).}

{Taking into account the cost of the Cholesky factorization, the overall computational complexity of our approach scales roughly as $O(n^3)$. However, it is important to note that the IPM matrices \eqref{eq:schur_IPM} in this application are sparse, and the corresponding Cholesky factors remain relatively sparse when standard reordering techniques are employed—this reordering is automatically performed by Matlab's \texttt{chol} function. The efficiency of sparse Cholesky factorization is heavily influenced by the amount of fill-in introduced during the factorization process. By minimizing fill-in through effective reordering, the practical computational cost and memory requirements can be significantly reduced compared to the worst-case $O(n^3)$ scenario, thereby enhancing the method's applicability to large-scale problems.}

Finally, concerning the IPM-related hyperparameters used in our numerical experimentation we used $\rho=\delta=10^{-10}$ and  $\texttt{StopTol}=10^{-8}$. Such a high level of accuracy is necessary when computing perturbations able to match the relative order in the target centralities and motivates the use of second-order methods in this particular context.

\subsection{Enforcing Katz centrality}
In this section, we first test the solution algorithm on a small-scale road network for which we can analyze in detail the obtained perturbation (Section~\ref{sec:small-net}), then we focus in Section~\ref{sec:Katz_S1} and~\ref{sec:Katz_S2} on the two scenarios described in Section~\ref{sec:target-distribution}.

\subsubsection{Enforcing Katz centrality -- Small Networks}\label{sec:small-net}  We begin by testing the problem of determining the perturbation that allows us to fix the Katz centrality on a small-sized network to describe the characteristics of the solution obtained. 
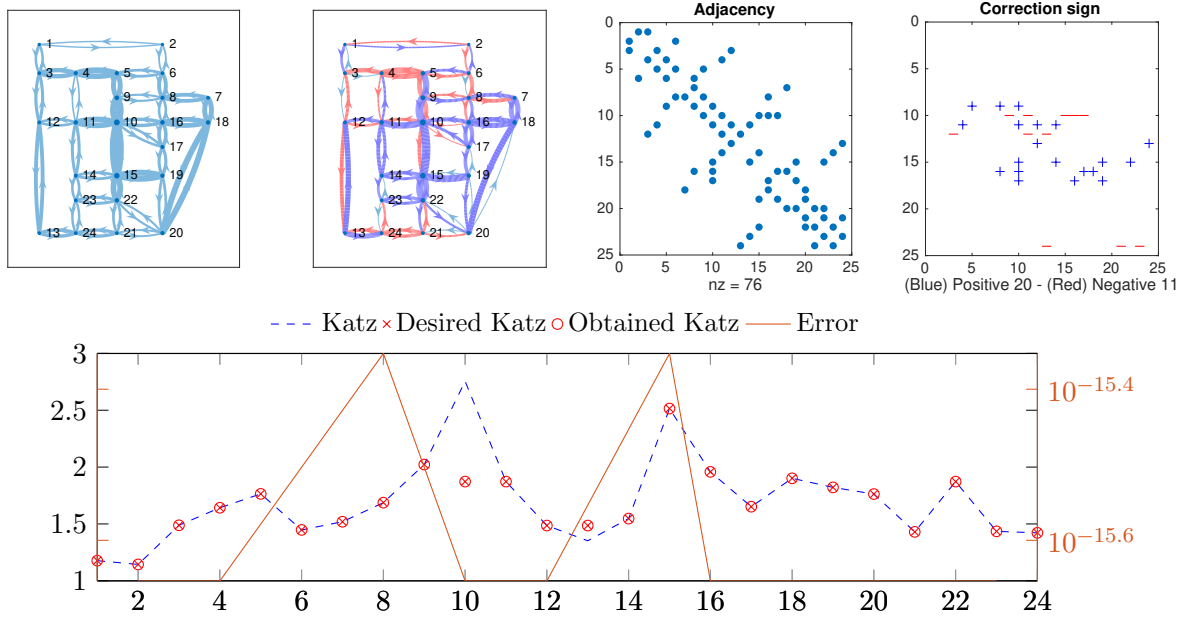
\begin{figure}[htbp]
    \centering
    \includegraphics[width=\columnwidth]{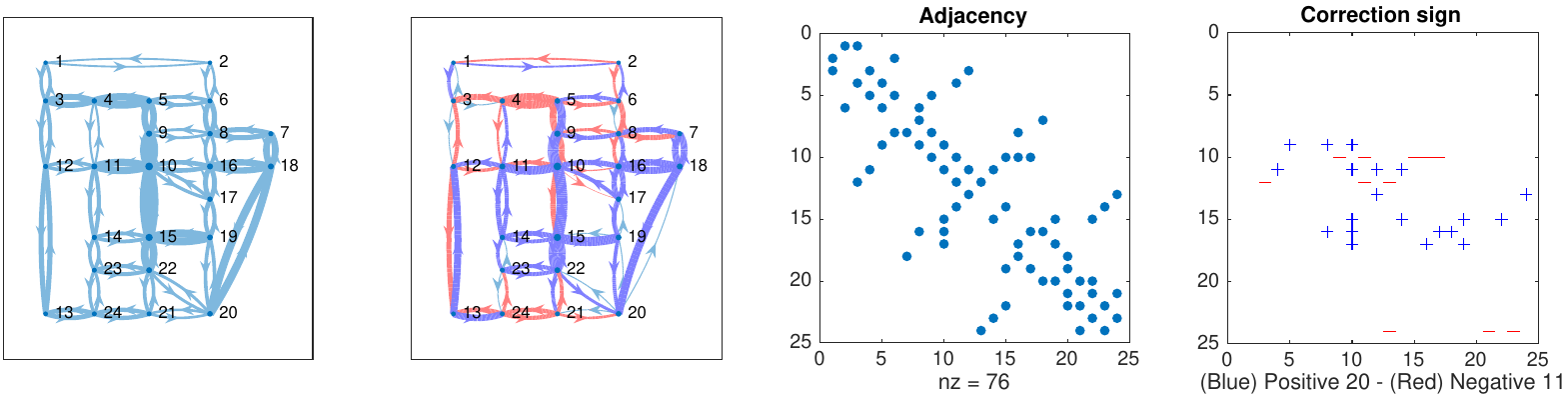}
    
    \definecolor{mycolor1}{rgb}{0.85000,0.32500,0.09800}%
\begin{tikzpicture}

\begin{axis}[%
width=0.8\columnwidth,
height=0.193\columnwidth,
at={(0\columnwidth,0\columnwidth)},
scale only axis,
xmin=1,
xmax=24,
separate axis lines,
every outer y axis line/.append style={black},
every y tick label/.append style={font=\color{black}},
every y tick/.append style={black},
ymin=1,
ymax=3,
yminorticks=true,
axis background/.style={fill=white},
yticklabel pos=left,
legend style={at={(0.5,1.03)}, anchor=south, legend cell align=left, align=left, draw=white!15!black}
]
\addplot [color=blue, dashed]
  table[row sep=crcr]{%
1	1.17726860009417\\
2	1.14366243402819\\
3	1.48809353365516\\
4	1.64210655512503\\
5	1.76428211796041\\
6	1.44741013474019\\
7	1.51922950764356\\
8	1.6871873600209\\
9	2.02208751431829\\
10	2.75454305312469\\
11	1.87230742507343\\
12	1.48574495887839\\
13	1.35205732907686\\
14	1.54709240180482\\
15	2.51471376538055\\
16	1.95774461018643\\
17	1.65169706468625\\
18	1.90120060948081\\
19	1.82146583163914\\
20	1.76264621435344\\
21	1.43019589826686\\
22	1.87231527612034\\
23	1.4348501357243\\
24	1.42178761327954\\
};
\label{Katz_qp}

\addplot [color=red, only marks, mark=x, mark options={solid, red}]
  table[row sep=crcr]{%
1	1.17726860009417\\
2	1.14366243402819\\
3	1.48809353365516\\
4	1.64210655512503\\
5	1.76428211796041\\
6	1.44741013474019\\
7	1.51922950764356\\
8	1.6871873600209\\
9	2.02208751431829\\
10	1.87230742507343\\
11	1.87230742507343\\
12	1.48574495887839\\
13	1.48574495887839\\
14	1.54709240180482\\
15	2.51471376538055\\
16	1.95774461018643\\
17	1.65169706468625\\
18	1.90120060948081\\
19	1.82146583163914\\
20	1.76264621435344\\
21	1.43019589826686\\
22	1.87231527612034\\
23	1.4348501357243\\
24	1.42178761327954\\
};
\label{Desired_Katz_qp}

\addplot [color=red, only marks, mark=o, mark options={solid, red}]
  table[row sep=crcr]{%
1	1.17726860009417\\
2	1.14366243402819\\
3	1.48809353365516\\
4	1.64210655512503\\
5	1.76428211796041\\
6	1.44741013474019\\
7	1.51922950764356\\
8	1.6871873600209\\
9	2.02208751431829\\
10	1.87230742507343\\
11	1.87230742507343\\
12	1.48574495887839\\
13	1.48574495887839\\
14	1.54709240180482\\
15	2.51471376538055\\
16	1.95774461018643\\
17	1.65169706468625\\
18	1.90120060948081\\
19	1.82146583163914\\
20	1.76264621435344\\
21	1.43019589826686\\
22	1.87231527612034\\
23	1.4348501357243\\
24	1.42178761327954\\
};
\label{Obtained_Katz_qp}
\end{axis}
\begin{axis}[%
width=0.8\columnwidth,
height=0.193\columnwidth,
at={(0\columnwidth,0\columnwidth)},
scale only axis,
xmin=1,
xmax=24,
separate axis lines,
every outer y axis line/.append style={mycolor1},
every y tick label/.append style={font=\color{mycolor1}},
every y tick/.append style={mycolor1},
ymode=log,
ymin=2.22044604925032e-16,
ymax=4.44089209850063e-16,
yminorticks=true,
axis background/.style={fill=none},
yticklabel pos=right,
legend columns=4,
legend style={at={(0.5,1.03)}, anchor=south, legend cell align=left, align=left, draw=none, font={\small}}
]
\addlegendimage{/pgfplots/refstyle=Katz_qp}\addlegendentry{Katz}
\addlegendimage{/pgfplots/refstyle=Desired_Katz_qp}\addlegendentry{Desired Katz}
\addlegendimage{/pgfplots/refstyle=Obtained_Katz_qp}\addlegendentry{Obtained Katz}
\addplot [color=mycolor1]
  table[row sep=crcr]{%
1	0\\
2	2.22044604925031e-16\\
3	0\\
4	2.22044604925031e-16\\
5	0\\
6	0\\
7	0\\
8	4.44089209850063e-16\\
9	0\\
10	2.22044604925031e-16\\
11	0\\
12	2.22044604925031e-16\\
13	0\\
14	0\\
15	4.44089209850063e-16\\
16	2.22044604925031e-16\\
17	0\\
18	0\\
19	0\\
20	2.22044604925031e-16\\
21	2.22044604925031e-16\\
22	0\\
23	2.22044604925031e-16\\
24	0\\
};
\addlegendentry{Error}

\end{axis}

\end{tikzpicture}%
    
    \caption{Sioux Falls road network (\href{https://tzin.bgu.ac.il/~bargera/tntp/}{tzin.bgu.ac.il/$\sim$bargera/tntp}). The first and third panels depict the road network with the directed edges with the corresponding weights together with the adjacency matrix. The third and fourth panels depict the modification needed to get the desired score vector for $\beta=1$. The blue ``$+$'' represents edges whose weight has been increased, and the red ``$-$'' edges whose weight has been decreased. On the second row, the graph reports the original score vector, dashed blue line, the desired Katz score, red crosses, and the one obtained through the optimization, red circles. On the right $y$-axis the error between the desired and the obtained is given.}
    \label{fig:sioux-falls}
\end{figure}
We consider a small road network, the network of the city of Sioux Falls (24 nodes, 76 links), and compute the true ranking $\boldsymbol{\mu}$ for $\alpha = \nicefrac{1}{2\rho(A)}$. The target score vector $\widehat{\boldsymbol{\mu}}$ coincides with $\boldsymbol{\mu}$ a part from $\widehat{\boldsymbol{\mu}}_{10} = \boldsymbol{\mu}_{11}$ and $\widehat{\boldsymbol{\mu}}_{13} = \boldsymbol{\mu}_{12}$. The modification pattern is that of the adjacency matrix $A$.
Figure~\ref{fig:sioux-falls} reports the results for this configuration for $\beta=1$. The algorithm modifies 31 edges (20 have an increase of their weight, while 11 are decreased) and {produces} a perturbation $\Delta$ such that $\|\Delta\|_F \approx  4.5$ {($ \nicefrac{\|\Delta\|_F}{\|A\|_F} \approx 0.20$)}. This is sufficient to have an absolute error with respect to the desired Katz vector falling under machine precision--see the bottom of Figure~\ref{fig:sioux-falls}. If we set $\beta = 0.2$ we manage to reduce the number of modified nodes from 31 to 29 decreasing by two the number of negative corrections, and moving the norm of the correction to $\|\Delta\|_F \approx 4.58$ {($ \nicefrac{\|\Delta\|_F}{\|A\|_F} \approx 0.21$)}.

\subsubsection{The S1 scenario} \label{sec:Katz_S1} 
In this section, we analyse the S1 scenario for enforcing Katz centrality. To showcase the robustness of our approach, we start analysing an extension of the S1 scenario, namely, the case where given an initial Katz ranking, we seek a perturbation matrix that makes the first $x = 10,20,30,40,50\%$ of the nodes in the ranking indistinguishable. Again, we choose to fix all nodes to a ranking value equal to their mean, i.e., some vertices will have their scores reduced, while others will have them increased. We test the procedure for different values of the $\beta$ parameter in~\eqref{eq:Katz_1_alpha_beta}, and for $\alpha = \nicefrac{1}{2 \rho(A)}$. In all cases, the modification pattern chosen is that of the matrix $A$ itself.
\begin{figure}[htbp]
    \centering
    \includegraphics[width=0.95\columnwidth]{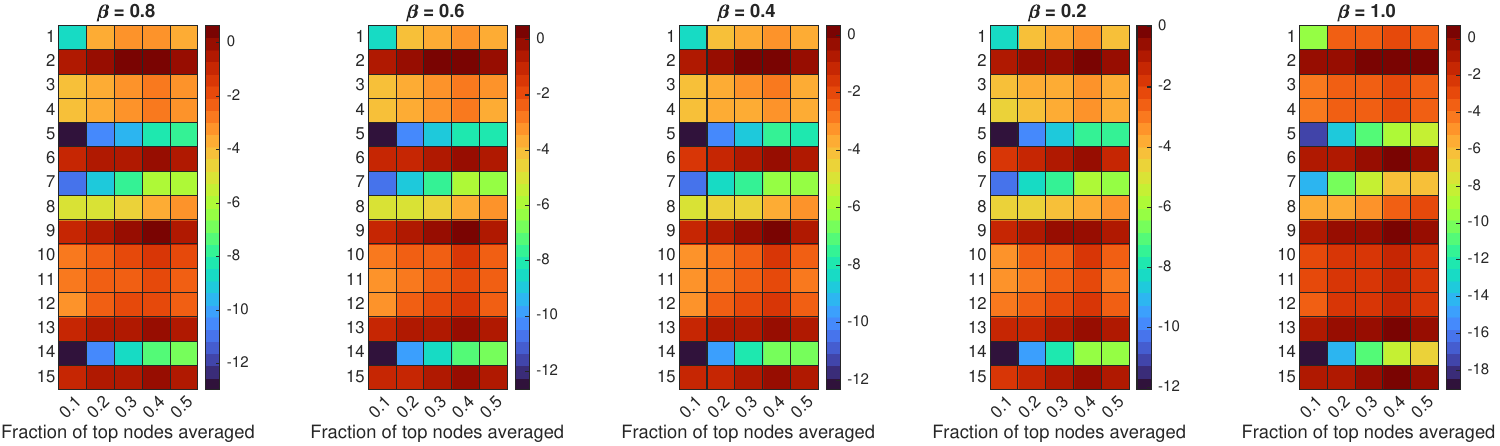}
    \caption{S1 Scenario. Value in $\log_{10}$-scale of the relative objective function $\nicefrac{J(\Delta)}{J(A)}$ in~\eqref{eq:Katz_1_alpha_beta}. On the columns, we read the fraction of equalized vertices in increasing order, on the rows, the different test cases as numbered  in~Table~\ref{tab:dataset}.}
    \label{fig:objective_function}
\end{figure}
 In Figure~\ref{fig:objective_function} we report the relative value of the objective function, from which we can observe that, generally, such value -- hence the norm of the computed perturbation $\Delta$ -- grows with the importance of the changes required on the centrality vector. 
\begin{figure}[htbp]
    \centering
    \includegraphics[width=0.95\columnwidth]{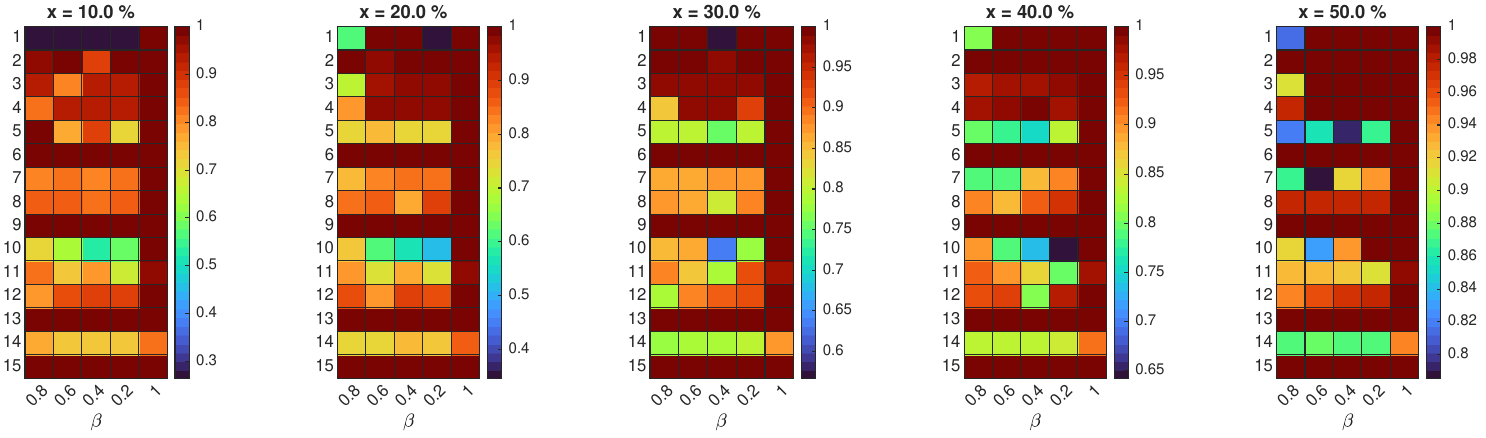}
    \caption{S1 Scenario. Number of nonzero entries obtained by solving~\eqref{eq:Katz_1_alpha_beta} scaled by the number of nonzero entries of the original adjacency matrix. On the columns we read the value of the $\beta$ parameter, on the rows, the different test cases as numbered in~Table~\ref{tab:dataset}. Each block is obtained for a different percentage of the averaged nodes.}
    \label{fig:nnz}
\end{figure}
In Figure~\ref{fig:nnz} we analyse, instead, the effect of increasing the weight of the \emph{sparsifying} term $\|\Delta\|_1$ in the objective function~\eqref{eq:Katz_1_alpha_beta} concerning the term $\|\Delta\|_F$. The presented results confirm in several cases the intuition that increasing values of $\beta$ correspond perturbations $\Delta$s with fewer numbers of modified links. In some cases, the solution that minimizes the objective function still has exactly zero elements compared to the number of elements of the required pattern. The experimental results presented above suggest that, when the choice of editing pattern is not clear, it is possible to select a larger modification pattern and exploit the $\beta$ parameter and the sparsification strategy here presented to produce a modification pattern with fewer non-zero elements with respect to the originally selected one.
Finally, it is important to note that when Kendall's $\kappa_\tau$ is considered as a measure of the goodness of the obtained ranking, for all the cases presented above, we always obtain $\kappa_\tau=1$ (we compare the first six decimal places of the values of the obtained ranking vector with the desired one). That is, modulo using reasonable precision, we can achieve, in all cases, exactly the desired relative ordering of the node scores. See also the experimental results presented in the next section.

\paragraph{Detailed results} 
In Table~\ref{tab:S2-tab1} we present the results for $\beta = 1$, and $\alpha = \nicefrac{1}{2\rho(A)}$ for the S1 scenario.  The details here reported aim at showcasing both the  overall quality of the obtained perturbation (in terms of relative Frobenius norm, sparsity and Kendall's correlation) and the performance of the IPM solver here considered (in terms of computational time and number of IPM iterations).

\begin{table}[htbp]\centering
\caption{Katz. S1 Scenario. $\beta=1$. The numbering of the test cases corresponds to the matrices in Table~\ref{tab:dataset}. The table reports:  the conditioning of the constraints matrix, the number of IPM iterations,
the elapsed time measured in seconds, the relative norm of the obtained perturbation together with the number of nonzero entries of positive ($+$), negative $(-)$ sign and the number of non-zero elements in the Cholesky factor of equation \eqref{eq:schur_IPM} in the last IPM iteration. The correlation between the obtained
centrality measure with the target one is measured in terms of Kendall's $\kappa_\tau$.}\label{tab:S2-tab1}
\begin{tabular}{rllllrrrc}
\toprule
& $\operatorname{cond}(LL^{\top})$ & Iter & T (s) & $\nicefrac{\|\Delta\|_F}{\|A\|_F}$ & $+$  & $-$ & \texttt{nnz Chol.} & $\kappa_\tau$ \\
&     &      &       &                &      &     & of \eqref{eq:schur_IPM} &  \\
\midrule
 1 ]  & 1.186e+00                                       & 5          & 1.09    & 1.120e-05        & 124858        & 129762          & 6545  &   1.00 \\   
 2 ]  & 3.692e+02                                       & 8          & 0.31    & 5.163e-01        & 17373        & 20353          & 10680  &   1.00 \\   
 3 ]  & 5.296e+00                                       & 8          & 0.70    & 6.130e-03        & 54089        & 86199          & 11397  &   1.00 \\   
 4 ]  & 1.067e+01                                       & 9          & 2.81    & 5.845e-03        & 114007        & 188242          & 39082  &   1.00 \\   
 5 ]  & 2.169e+00                                       & 5          & 0.32    & 2.369e-09        & 2564        & 9004          & 22499  &   1.00 \\   
 6 ]  & 5.569e+01                                       & 20          & 3.86    & 2.897e-01        & 128611        & 152823          & 67578  &   1.00 \\   
 7 ]  & 3.447e+00                                       & 5          & 0.33    & 9.861e-08        & 3072        & 11088          & 16840  &   1.00 \\   
 8 ]  & 1.037e+01                                       & 5          & 0.12    & 1.641e-03        & 4164        & 1781          & 2851  &   1.00 \\   
 9 ]  & 4.784e+01                                       & 19          & 1.34    & 4.129e-01        & 44445        & 55486          & 24115  &   1.00 \\   
 10 ]  & 8.285e+00                                       & 5          & 1.20    & 3.468e-02        & 118601        & 139384          & 65536  &   1.00 \\   
 11 ]  & 5.542e+00                                       & 5          & 0.26    & 3.166e-02        & 12583        & 13393          & 11143  &   1.00 \\   
 12 ]  & 6.229e+00                                       & 10          & 0.08    & 2.012e-02        & 1177        & 2518          & 1107  &   1.00 \\   
 13 ]  & 7.981e+01                                       & 23          & 2.98    & 3.437e-01        & 98462        & 115140          & 48837  &   1.00 \\   
 14 ]  & 5.676e+00                                       & 5          & 0.08    & 2.687e-10        & 386        & 1409          & 4824  &   1.00 \\   
 15 ]  & 8.013e+01                                       & 20          & 1.89    & 3.076e-01        & 63052        & 70549          & 32580  &   1.00 \\  
\bottomrule
\end{tabular}
\end{table}

In Table~\ref{tab:S2-tab2} we report the analogous experiment in which we vary the value of the parameter $\beta$, that is, the reciprocal weight between the term in the Frobenius norm and the one in the $1$--norm.
\begin{table}[htbp]\centering
\caption{Katz. S1 Scenario. $(1-\beta)/\beta=100$. The numbering of the test cases corresponds to the matrices in Table~\ref{tab:dataset}. The table reports:  the conditioning of the constraints matrix, the number of IPM iterations,
the elapsed time measured in seconds, the relative norm of the obtained perturbation together with the number of nonzero entries of positive ($+$), negative $(-)$ sign and the number of non-zero elements in the Cholesky factor of equation \eqref{eq:schur_IPM} in the last IPM iteration. The correlation between the obtained
centrality measure with the target one is measured in terms of Kendall's $\kappa_\tau$.}\label{tab:S2-tab2}
\begin{tabular}{rllllrrrcc}
\toprule
& $\operatorname{cond}(LL^{\top})$ & Iter & T (s) & $\nicefrac{\|\Delta\|_F}{\|A\|_F}$ & $+$  & $-$ & \texttt{nnz Chol.} & $\kappa_\tau$ \\
&     &      &       &                &      &     & of \eqref{eq:schur_IPM} &  \\
\midrule
 1 ]  & 1.500e+02                                        & 13        & 9.21    & 1.802e-02         & 19229        & 23006          & 597905   &  1.00 \\  
 2 ]  & 1.282e+03                                        & 15        & 1.81    & 1.591e+00         & 10811        & 13682          & 107944   &  1.00 \\  
 3 ]  & 9.011e+01                                        & 17        & 5.73    & 1.907e-02         & 6419        & 7714          & 312687   &  1.00 \\  
 4 ]  & 1.183e+02                                        & 18        & 24.76    & 1.901e-02         & 23244        & 30247          & 1158526   &  1.00 \\  
 5 ]  & 1.522e+01                                        & 8        & 1.69    & 6.483e-02         & 2564        & 9004          & 197931   &  1.00 \\  
 6 ]  & 2.195e+02                                        & 28        & 21.42    & 3.822e-01         & 126371        & 150527          & 740282   &  1.00 \\  
 7 ]  & 2.212e+01                                        & 8        & 1.72    & 2.978e-02         & 3072        & 11088          & 209768   &  1.00 \\  
 8 ]  & 5.360e+01                                        & 13        & 1.02    & 8.773e-02         & 3716        & 1720          & 63223   &  1.00 \\  
 9 ]  & 1.880e+02                                        & 27        & 7.26    & 5.465e-01         & 43667        & 54587          & 256227   &  1.00 \\  
 10 ]  & 6.455e+01                                        & 15        & 14.35    & 1.015e-01         & 27299        & 32691          & 851836   &  1.00 \\  
 11 ]  & 4.464e+01                                        & 13        & 1.97    & 1.247e-01         & 4427        & 4300          & 142415   &  1.00 \\  
 12 ]  & 3.598e+01                                        & 11        & 0.22    & 9.505e-02         & 193        & 202          & 12435   &  1.00 \\  
 13 ]  & 3.129e+02                                        & 28        & 14.43    & 4.453e-01         & 94368        & 110668          & 517937   &  1.00 \\  
 14 ]  & 1.896e+01                                        & 7        & 0.25    & 4.968e-02         & 366        & 1389          & 32172   &  1.00 \\  
 15 ]  & 2.944e+02                                        & 26        & 8.34    & 5.337e-01         & 61809        & 69315          & 343776   &  1.00 \\  
\bottomrule
\end{tabular}
\end{table} 
Specifically we select a value of $\beta$ such that $\nicefrac{(1-\beta)}{\beta}=100$. Comparing these results with those in Table~\ref{tab:S2-tab1}, we observe as expected, a decrease in the number of non-zero entries of the computed perturbation at the price of an increased relative norm $\|\Delta\|_F/\|A\|_F$ due to the presence in the objective function also of the $1$-norm  penalty. This makes the optimization procedure more complicated, as observed by the growth in the number of iterations needed to reach convergence.

\subsubsection{The S2 scenario} \label{sec:Katz_S2} 
In Tables~\ref{tab:s3-tab1} and~\ref{tab:s3-tab2} we report the result obtained for the S2 scenario. The tables are relative to the same choices of the parameters $\beta$ and $\alpha$ ($\beta = 1$ and of $\alpha = \nicefrac{1}{2\rho(A)}$)
we employed the S1 scenario and reported the same metrics used there.
\begin{table}[htbp]\centering
\caption{Katz. S2 Scenario. $\beta=1$. The numbering of the test cases corresponds to the matrices in Table~\ref{tab:dataset}. The table reports:  the conditioning of the constraints matrix, the number of IPM iterations,
the elapsed time measured in seconds, the relative norm of the obtained perturbation together with the number of nonzero entries of positive ($+$), negative $(-)$ sign and the number of non-zero elements in the Cholesky factor of equation \eqref{eq:schur_IPM} in the last IPM iteration. The correlation between the obtained
centrality measure with the target one is measured in terms of Kendall's $\kappa_\tau$.}\label{tab:s3-tab1}
\begin{tabular}{rllllrrrcc}
\toprule
& $\operatorname{cond}(LL^{\top})$ & Iter & T (s) & $\nicefrac{\|\Delta\|_F}{\|A\|_F}$ & $+$  & $-$ & \texttt{nnz Chol.} & $\kappa_\tau$ \\
&     &      &       &                &      &     & of \eqref{eq:schur_IPM} &  \\

\midrule
 1 ]  & 1.199e+00                                           & 5          & 0.94    & 1.802e-02      & 146886        & 101416          & 6545  &   1.00  \\   
 2 ]  & 5.889e+02                                           & 15          & 0.63    & 1.591e+00      & 24314        & 7828          & 10680  &   1.00  \\   
 3 ]  & 5.673e+00                                           & 12          & 1.09    & 1.907e-02      & 9923        & 7698          & 11397  &   1.00  \\   
 4 ]  & 1.137e+01                                           & 12          & 3.78    & 1.901e-02      & 33462        & 31160          & 39082  &   1.00  \\   
 5 ]  & 2.169e+00                                           & 5          & 0.32    & 6.483e-02      & 17526        & 14068          & 22499  &   1.00  \\   
 6 ]  & 5.723e+01                                           & 24          & 4.26    & 3.822e-01      & 111141        & 116710          & 67578  &   1.00  \\   
 7 ]  & 3.447e+00                                           & 5          & 0.36    & 2.978e-02      & 8688        & 5453          & 16840  &   1.00  \\   
 8 ]  & 1.037e+01                                           & 5          & 0.12    & 8.773e-02      & 2013        & 1578          & 2851  &   1.00  \\   
 9 ]  & 4.786e+01                                           & 21          & 1.56    & 5.465e-01      & 42497        & 43350          & 24115  &   1.00  \\   
 10 ]  & 1.169e+01                                           & 6          & 1.43    & 1.015e-01      & 100320        & 85616          & 65536  &   1.00  \\   
 11 ]  & 6.158e+00                                           & 5          & 0.23    & 1.247e-01      & 7251        & 6783          & 11143  &   1.00  \\   
 12 ]  & 7.461e+00                                           & 10          & 0.07    & 9.505e-02      & 991        & 2714          & 1107  &   1.00  \\   
 13 ]  & 8.208e+01                                           & 23          & 2.95    & 4.453e-01      & 89297        & 92063          & 48837  &   1.00  \\   
 14 ]  & 5.676e+00                                           & 6          & 0.08    & 4.968e-02      & 1002        & 669          & 4824  &   1.00  \\   
 15 ]  & 8.338e+01                                           & 21          & 1.91    & 5.337e-01      & 46689        & 50026          & 32580  &   1.00  \\   
\bottomrule
\end{tabular}
\end{table} 

\begin{table}[htbp]\centering
\caption{Katz. S2 Scenario. $(1-\beta)/\beta=100$. The numbering of the test cases corresponds to the matrices in Table~\ref{tab:dataset}. The table reports:  the conditioning of the constraints matrix, the number of IPM iterations,
the elapsed time measured in seconds, the relative norm of the obtained perturbation together with the number of nonzero entries of positive ($+$), negative $(-)$ sign and the number of non-zero elements in the Cholesky factor of equation \eqref{eq:schur_IPM} in the last IPM iteration. The correlation between the obtained
centrality measure with the target one is measured in terms of Kendall's $\kappa_\tau$.}\label{tab:s3-tab2}
\begin{tabular}{rllllrrrcc}
\toprule
& $\operatorname{cond}(LL^{\top})$ & Iter & T (s) & $\nicefrac{\|\Delta\|_F}{\|A\|_F}$ & $+$  & $-$ & \texttt{nnz Chol.} & $\kappa_\tau$ \\
&     &      &       &                &      &     & of \eqref{eq:schur_IPM} &  \\

\midrule
 1 ]  & 1.512e+02                                         & 22       & 15.23    & 2.485e+01      & 26476        & 18341          & 597905  &   1.00 \\   
 2 ]  & 1.838e+03                                         & 24       & 2.84    & 3.576e+02      & 2994        & 4411          & 107944  &   1.00 \\   
 3 ]  & 9.129e+01                                         & 17       & 5.45    & 2.678e+00      & 1330        & 2041          & 312687  &   1.00 \\   
 4 ]  & 1.176e+02                                         & 20       & 28.98    & 5.288e+00      & 3961        & 7646          & 1158526  &   1.00 \\   
 5 ]  & 1.522e+01                                         & 18       & 3.72    & 2.588e+01      & 10566        & 8091          & 197931  &   1.00 \\   
 6 ]  & 2.242e+02                                         & 32       & 23.96    & 1.251e+07      & 84844        & 86183          & 740282  &   1.00 \\   
 7 ]  & 2.212e+01                                         & 17       & 3.88    & 1.288e+01      & 7671        & 4061          & 209768  &   1.00 \\   
 8 ]  & 5.358e+01                                         & 14       & 1.05    & 2.125e+01      & 743        & 637          & 63223  &   1.00 \\   
 9 ]  & 1.884e+02                                         & 31       & 8.02    & 1.244e+07      & 35519        & 36245          & 256227  &   1.00 \\   
 10 ]  & 8.551e+01                                         & 18       & 17.32    & 1.073e+02      & 25598        & 23428          & 851836  &   1.00 \\   
 11 ]  & 4.505e+01                                         & 17       & 2.79    & 4.766e+01      & 3149        & 2935          & 142415  &   1.00 \\   
 12 ]  & 4.178e+01                                         & 15       & 0.26    & 2.683e+00      & 116        & 167          & 12435  &   1.00 \\   
 13 ]  & 3.151e+02                                         & 33       & 16.85    & 1.894e+07      & 60963        & 61177          & 517937  &   1.00 \\   
 14 ]  & 1.896e+01                                         & 16       & 0.65    & 6.828e+00      & 799        & 541          & 32172  &   1.00 \\   
 15 ]  & 3.026e+02                                         & 29       & 9.28    & 2.274e+07      & 25824        & 27907          & 343776  &   1.00 \\  
\bottomrule
\end{tabular}
\end{table}

The comparison of Tables~\ref{tab:S2-tab1} and~\ref{tab:S2-tab2} (S1) with Tables~\ref{tab:s3-tab1} and~\ref{tab:s3-tab2} (S2) confirm the robustness of the approach here proposed for different target distributions, producing in both cases very high correlation coefficients. The same comparison shows, moreover, that in the S2 scenario, in general, perturbations with larger relative norms have to be produced and that times and number of IPM iterations remain limited and comparable with those from the S1 scenario.

\subsection{Enforcing PageRank centrality}\label{sec:enforcing_pagerank}
In this section we describe some experiments related to the solution of the problem $\mathcal{P}_{\alpha,\beta}^{\text{Pr}}$ as formulated in~\eqref{eq:p1PR_standard_form_Final}. Specifically, in Section~\ref{sec:tsdp_comparison} we compare our proposal with the method introduced in~\cite{gillis2024assigning} for solving the target stationary distribution problem with a sparse stochastic matrix. In Section~\ref{sec:PR_S1_S2} we report the numerical results for the solution of the Pagerank centrality assignment when $\alpha= 0.9$ and $\mathbf{t}=\frac{1}{n}\mathbf{1}$ for the dataset in Table~\ref{tab:dataset}.   

\subsubsection{Comparison with the TSDP method}\label{sec:tsdp_comparison}

The method introduced in~\cite{gillis2024assigning} addresses a problem similar to $\mathcal{P}_{\alpha,\beta}^{\text{Pr}}$ as formulated in~\eqref{eq:p1PR_standard_form_Final}, but with a significant difference: it identifies a sparse perturbation of a sparse stochastic matrix to assign a predetermined stationary vector. In contrast, our approach preserves the structure of the PageRank problem by modifying only the sparse part of the matrix and the $\alpha$ value, while keeping the preference vector unchanged.
{In particular, in the TSDP method, we look for a modification $\Delta_{\text{TDSP}}$ such that $G+\Delta_{\text{TDSP}}$ has the desired stationary distribution $\widehat{\boldsymbol{\pi}}$, where $G$ is defined in \eqref{eq:G}. The resulting modified stochastic matrix is $\widetilde{G}=G+\Delta_{\text{TDSP}}=\alpha (D^{-1}A) +(1-\alpha) \mathbf{1} \mathbf{t}^{\top}+\Delta_{\text{TDSP}}$. By construction $( \alpha (D^{-1}A) +\Delta_{\text{TDSP}})\mathbf{1}=\alpha \mathbf{1}$. However, $ \alpha (D^{-1}A) +\Delta_{\text{TDSP}}$ might have negative entries, even though $\Delta_{\text{TDSP}}$ is chosen to have the same sparsity pattern of $I+A$. To this regard, we refer to Example~\ref{example:negative}. In our approach, we circumvent this issue by acting on the teleportation parameter $\alpha$ (see Proposition~\ref{pro:interpretability}). }

To illustrate the difference between TSDP and our approach, we consider Zachary's karate club network (34 nodes, 156 edges). We aim to adjust the stationary vector so that nodes $1$ and $34$ have their weights reduced by $0.5$. Specifically, we set $\widehat{\boldsymbol{\pi}} = \boldsymbol{\pi}$, multiply the entries for nodes 1 and 34 by 0.5, and normalize $\widehat{\boldsymbol{\pi}}$ so that $\widehat{\boldsymbol{\pi}}^{\top}\mathbf{1} = 1$. We then apply our method ($\alpha = 0.85$, $\beta = 0.1$) to determine the perturbation $\Delta$. We also apply the TSDP method from~\cite{gillis2024assigning}, once using the pattern of $A+I$ as in our method, and once allowing modification of the entire matrix $G$.

As depicted in Fig.~\ref{fig:pr_stationary}, all methods generate perturbations that produce stationary vectors matching the target within machine precision. 
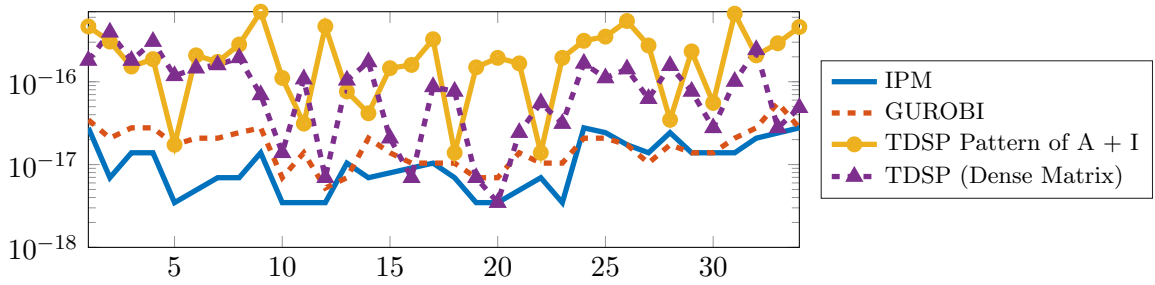
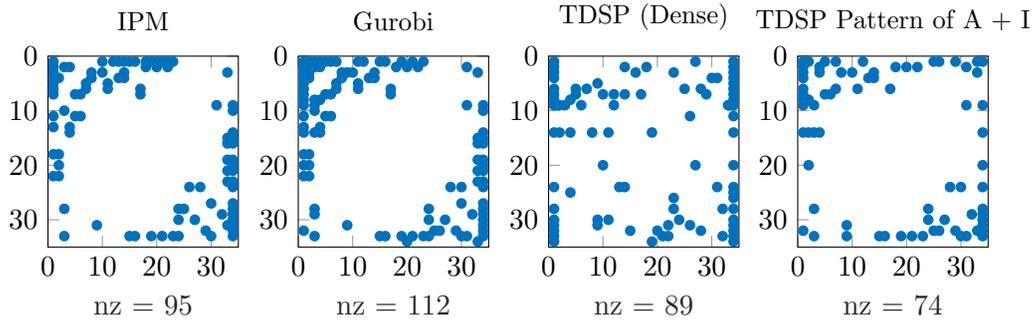
\begin{figure}[htbp]
    \centering
    \subfloat[Errors on the obtained stationary vectors with respect to the assigned stationary distribution.\label{fig:pr_stationary}]{\definecolor{mycolor1}{rgb}{0.00000,0.44700,0.74100}%
\definecolor{mycolor2}{rgb}{0.85000,0.32500,0.09800}%
\definecolor{mycolor3}{rgb}{0.92900,0.69400,0.12500}%
\definecolor{mycolor4}{rgb}{0.49400,0.18400,0.55600}%
\begin{tikzpicture}

\begin{axis}[%
width=0.605\columnwidth,
height=0.2\columnwidth,
at={(0\columnwidth,0\columnwidth)},
scale only axis,
xmin=1,
xmax=34,
ymode=log,
ymin=1e-18,
ymax=7.0082828429463e-16,
yminorticks=true,
axis background/.style={fill=white},
legend style={at={(1.03,0.5)}, anchor=west, legend cell align=left, align=left, draw=white!15!black, font=\footnotesize}
]
\addplot [color=mycolor1, line width=2.0pt]
  table[row sep=crcr]{%
1	2.77555756156289e-17\\
2	6.93889390390723e-18\\
3	1.38777878078145e-17\\
4	1.38777878078145e-17\\
5	3.46944695195361e-18\\
6	0\\
7	6.93889390390723e-18\\
8	6.93889390390723e-18\\
9	1.38777878078145e-17\\
10	3.46944695195361e-18\\
11	0\\
12	3.46944695195361e-18\\
13	1.04083408558608e-17\\
14	6.93889390390723e-18\\
15	0\\
16	0\\
17	1.04083408558608e-17\\
18	6.93889390390723e-18\\
19	3.46944695195361e-18\\
20	3.46944695195361e-18\\
21	0\\
22	6.93889390390723e-18\\
23	3.46944695195361e-18\\
24	2.77555756156289e-17\\
25	2.42861286636753e-17\\
26	1.73472347597681e-17\\
27	1.38777878078145e-17\\
28	2.42861286636753e-17\\
29	1.38777878078145e-17\\
30	1.38777878078145e-17\\
31	1.38777878078145e-17\\
32	2.08166817117217e-17\\
33	0\\
34	2.77555756156289e-17\\
};
\addlegendentry{IPM}

\addplot [color=mycolor2, dashed, line width=2.0pt]
  table[row sep=crcr]{%
1	3.46944695195361e-17\\
2	2.08166817117217e-17\\
3	2.77555756156289e-17\\
4	2.77555756156289e-17\\
5	1.73472347597681e-17\\
6	2.08166817117217e-17\\
7	2.08166817117217e-17\\
8	2.42861286636753e-17\\
9	2.77555756156289e-17\\
10	6.93889390390723e-18\\
11	1.38777878078145e-17\\
12	5.20417042793042e-18\\
13	6.93889390390723e-18\\
14	2.08166817117217e-17\\
15	1.38777878078145e-17\\
16	1.04083408558608e-17\\
17	1.04083408558608e-17\\
18	1.04083408558608e-17\\
19	6.93889390390723e-18\\
20	6.93889390390723e-18\\
21	1.38777878078145e-17\\
22	1.04083408558608e-17\\
23	1.04083408558608e-17\\
24	2.08166817117217e-17\\
25	2.08166817117217e-17\\
26	1.73472347597681e-17\\
27	1.04083408558608e-17\\
28	1.73472347597681e-17\\
29	1.38777878078145e-17\\
30	1.38777878078145e-17\\
31	2.08166817117217e-17\\
32	2.77555756156289e-17\\
33	5.55111512312578e-17\\
34	2.77555756156289e-17\\
};
\addlegendentry{GUROBI}

\addplot [color=mycolor3, line width=2.0pt, mark=o, mark options={solid, mycolor3}]
  table[row sep=crcr]{%
1	4.64905891561784e-16\\
2	3.05311331771918e-16\\
3	1.52655665885959e-16\\
4	1.87350135405495e-16\\
5	1.73472347597681e-17\\
6	2.08166817117217e-16\\
7	1.73472347597681e-16\\
8	2.81025203108243e-16\\
9	7.0082828429463e-16\\
10	1.11022302462516e-16\\
11	3.12250225675825e-17\\
12	4.66640615037761e-16\\
13	7.63278329429795e-17\\
14	4.16333634234434e-17\\
15	1.45716771982052e-16\\
16	1.59594559789866e-16\\
17	3.2612801348364e-16\\
18	1.38777878078145e-17\\
19	1.49186218934005e-16\\
20	1.94289029309402e-16\\
21	1.66533453693773e-16\\
22	1.38777878078145e-17\\
23	1.94289029309402e-16\\
24	3.12250225675825e-16\\
25	3.50414142147315e-16\\
26	5.41233724504764e-16\\
27	2.74086309204336e-16\\
28	3.46944695195361e-17\\
29	2.32452945780892e-16\\
30	5.55111512312578e-17\\
31	6.6266436782314e-16\\
32	2.08166817117217e-16\\
33	2.91433543964104e-16\\
34	4.57966997657877e-16\\
};
\addlegendentry{TDSP Pattern of A + I}

\addplot [color=mycolor4, dashed, line width=2.0pt, mark=triangle, mark options={solid, mycolor4}]
  table[row sep=crcr]{%
1	1.80411241501588e-16\\
2	3.95516952522712e-16\\
3	1.80411241501588e-16\\
4	3.05311331771918e-16\\
5	1.17961196366423e-16\\
6	1.45716771982052e-16\\
7	1.59594559789866e-16\\
8	1.94289029309402e-16\\
9	6.93889390390723e-17\\
10	1.38777878078145e-17\\
11	1.07552855510562e-16\\
12	6.93889390390723e-18\\
13	1.04083408558608e-16\\
14	1.73472347597681e-16\\
15	2.08166817117217e-17\\
16	6.93889390390723e-18\\
17	8.67361737988404e-17\\
18	7.63278329429795e-17\\
19	6.93889390390723e-18\\
20	3.46944695195361e-18\\
21	2.42861286636753e-17\\
22	5.55111512312578e-17\\
23	3.12250225675825e-17\\
24	1.66533453693773e-16\\
25	1.11022302462516e-16\\
26	1.42247325030098e-16\\
27	6.24500451351651e-17\\
28	1.56125112837913e-16\\
29	7.63278329429795e-17\\
30	2.77555756156289e-17\\
31	1.00613961606655e-16\\
32	2.42861286636753e-16\\
33	2.77555756156289e-17\\
34	4.85722573273506e-17\\
};
\addlegendentry{TDSP (Dense Matrix)}

\end{axis}

\end{tikzpicture}%
}
    
    \subfloat[Pattern of the different perturbation.\label{fig:pr_pattern}]{\definecolor{mycolor1}{rgb}{0.00000,0.44700,0.74100}%
\begin{tikzpicture}

\begin{axis}[%
width=0.162\columnwidth,
height=0.162\columnwidth,
at={(0\columnwidth,0\columnwidth)},
scale only axis,
xmin=0,
xmax=35,
xlabel style={font=\color{white!15!black}},
xlabel={nz = 95},
y dir=reverse,
ymin=0,
ymax=35,
axis background/.style={fill=white},
title style={font=\small},
title={IPM}
]
\addplot [color=mycolor1, line width=2.0pt, only marks, mark size=1pt, mark=*, mark options={solid, mycolor1}, forget plot]
  table[row sep=crcr]{%
1	1\\
1	2\\
1	3\\
1	4\\
1	5\\
1	6\\
1	7\\
1	11\\
1	13\\
1	18\\
1	22\\
2	4\\
2	18\\
2	20\\
2	22\\
3	2\\
3	10\\
3	28\\
3	33\\
4	2\\
4	13\\
4	14\\
5	7\\
5	11\\
6	7\\
6	11\\
7	5\\
7	6\\
8	3\\
8	4\\
9	31\\
10	1\\
10	3\\
11	5\\
11	6\\
12	1\\
13	1\\
13	4\\
14	1\\
14	3\\
14	4\\
15	1\\
15	33\\
16	1\\
16	33\\
17	6\\
17	7\\
18	1\\
18	2\\
19	1\\
19	33\\
20	1\\
20	2\\
21	1\\
21	33\\
22	1\\
22	2\\
23	1\\
23	33\\
24	28\\
24	30\\
24	33\\
25	28\\
26	24\\
27	30\\
28	24\\
29	32\\
30	27\\
30	33\\
31	9\\
32	29\\
33	3\\
33	15\\
33	16\\
33	19\\
33	21\\
33	23\\
33	31\\
34	9\\
34	10\\
34	14\\
34	15\\
34	16\\
34	19\\
34	20\\
34	21\\
34	23\\
34	24\\
34	27\\
34	28\\
34	29\\
34	30\\
34	31\\
34	32\\
34	33\\
};
\end{axis}

\begin{axis}[%
width=0.162\columnwidth,
height=0.162\columnwidth,
at={(0.213\columnwidth,0\columnwidth)},
scale only axis,
xmin=0,
xmax=35,
xlabel style={font=\color{white!15!black}},
xlabel={nz = 112},
y dir=reverse,
ymin=0,
ymax=35,
axis background/.style={fill=white},
title style={font=\small},
title={Gurobi}
]
\addplot [color=mycolor1, line width=2.0pt, only marks, mark size=1pt, mark=*, mark options={solid, mycolor1}, forget plot]
  table[row sep=crcr]{%
1	1\\
1	2\\
1	3\\
1	4\\
1	5\\
1	6\\
1	7\\
1	8\\
1	9\\
1	11\\
1	13\\
1	14\\
1	18\\
1	20\\
1	22\\
1	32\\
2	1\\
2	3\\
2	4\\
2	18\\
2	20\\
2	22\\
3	1\\
3	2\\
3	8\\
3	9\\
3	10\\
3	14\\
3	28\\
3	29\\
3	33\\
4	1\\
4	2\\
4	8\\
4	13\\
4	14\\
5	1\\
5	7\\
5	11\\
6	1\\
6	7\\
6	11\\
7	1\\
7	5\\
7	6\\
8	3\\
8	4\\
9	3\\
9	31\\
10	3\\
11	1\\
11	5\\
11	6\\
12	1\\
13	4\\
14	3\\
14	4\\
15	1\\
15	33\\
16	1\\
16	33\\
17	6\\
17	7\\
18	2\\
19	1\\
19	33\\
20	2\\
20	34\\
21	1\\
21	33\\
22	2\\
23	1\\
23	33\\
24	28\\
24	30\\
24	33\\
25	32\\
26	32\\
27	30\\
28	24\\
29	32\\
30	24\\
30	27\\
30	33\\
31	2\\
31	9\\
31	33\\
32	29\\
33	3\\
33	15\\
33	16\\
33	19\\
33	21\\
33	23\\
33	34\\
34	9\\
34	10\\
34	14\\
34	15\\
34	16\\
34	19\\
34	20\\
34	21\\
34	23\\
34	24\\
34	27\\
34	28\\
34	29\\
34	30\\
34	31\\
34	32\\
34	33\\
};
\end{axis}

\begin{axis}[%
width=0.162\columnwidth,
height=0.162\columnwidth,
at={(0.426\columnwidth,0\columnwidth)},
scale only axis,
xmin=0,
xmax=35,
xlabel style={font=\color{white!15!black}},
xlabel={nz = 89},
y dir=reverse,
ymin=0,
ymax=35,
axis background/.style={fill=white},
title style={font=\small},
title={TDSP (Dense)}
]
\addplot [color=mycolor1, line width=2.0pt, only marks, mark size=1pt, mark=*, mark options={solid, mycolor1}, forget plot]
  table[row sep=crcr]{%
1	1\\
1	2\\
1	3\\
1	4\\
1	6\\
1	7\\
1	8\\
1	9\\
1	14\\
1	24\\
1	28\\
1	30\\
1	31\\
1	32\\
1	33\\
1	34\\
2	9\\
2	14\\
3	9\\
4	8\\
4	14\\
4	25\\
5	6\\
5	7\\
6	9\\
7	6\\
8	14\\
9	5\\
9	30\\
9	31\\
10	7\\
10	20\\
11	14\\
11	30\\
12	7\\
12	9\\
13	24\\
14	2\\
14	7\\
14	24\\
15	32\\
16	3\\
17	7\\
18	2\\
19	14\\
19	34\\
20	32\\
21	33\\
22	31\\
22	33\\
23	3\\
23	26\\
23	28\\
24	30\\
25	6\\
26	11\\
26	31\\
27	1\\
27	20\\
28	6\\
28	32\\
29	7\\
30	4\\
31	4\\
31	24\\
32	28\\
32	33\\
33	9\\
34	1\\
34	2\\
34	3\\
34	4\\
34	5\\
34	6\\
34	7\\
34	8\\
34	9\\
34	11\\
34	14\\
34	20\\
34	24\\
34	25\\
34	26\\
34	28\\
34	30\\
34	31\\
34	32\\
34	33\\
34	34\\
};
\end{axis}

\begin{axis}[%
width=0.162\columnwidth,
height=0.162\columnwidth,
at={(0.638\columnwidth,0\columnwidth)},
scale only axis,
xmin=0,
xmax=35,
xlabel style={font=\color{white!15!black}},
xlabel={nz = 74},
y dir=reverse,
ymin=0,
ymax=35,
axis background/.style={fill=white},
title style={font=\small},
title={TDSP Pattern of A + I}
]
\addplot [color=mycolor1, line width=2.0pt, only marks, mark size=1pt, mark=*, mark options={solid, mycolor1}, forget plot]
  table[row sep=crcr]{%
1	1\\
1	2\\
1	3\\
1	4\\
1	6\\
1	7\\
1	8\\
1	9\\
1	14\\
1	32\\
2	1\\
2	8\\
2	14\\
2	20\\
3	9\\
3	14\\
3	28\\
3	33\\
4	2\\
4	3\\
4	14\\
5	1\\
5	7\\
6	7\\
7	1\\
7	6\\
8	3\\
9	31\\
9	33\\
10	3\\
11	6\\
12	1\\
13	4\\
14	3\\
14	4\\
15	33\\
16	33\\
17	6\\
18	2\\
19	33\\
20	2\\
21	33\\
22	2\\
23	33\\
24	28\\
24	30\\
25	1\\
25	32\\
26	1\\
26	32\\
27	30\\
28	24\\
29	1\\
29	32\\
30	24\\
30	33\\
31	2\\
31	9\\
31	33\\
32	29\\
33	1\\
33	3\\
33	32\\
34	1\\
34	9\\
34	14\\
34	20\\
34	24\\
34	28\\
34	29\\
34	30\\
34	31\\
34	32\\
34	33\\
};
\end{axis}
\end{tikzpicture}%
}
    
    \caption{Enforcing PageRank. Error on the target stationary distribution for the different methods and comparison of the perturbations patterns.}
\end{figure}
Examining the perturbation patterns in Fig.~\ref{fig:pr_pattern}, we see that the TDSP algorithm alters fewer entries in the dense probability matrix for both pattern choices. However, restricting changes to the sparse part of the matrix results in more non-zero elements. 
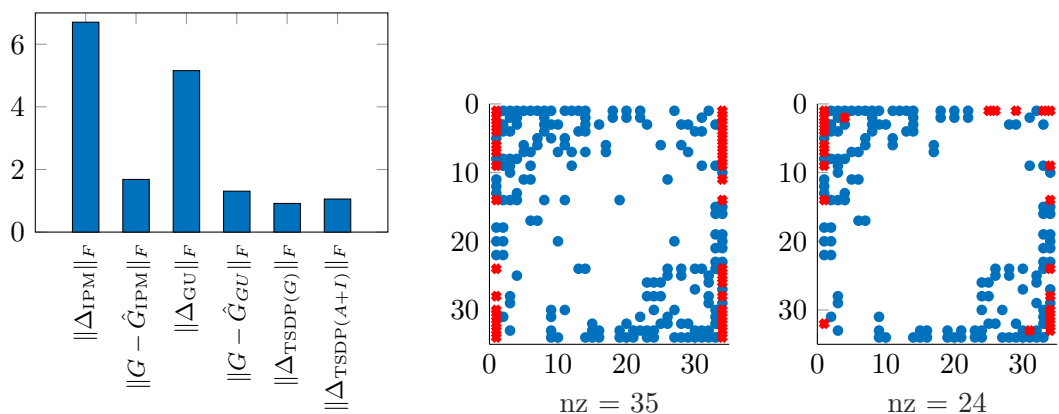
\begin{figure}[htbp]
    \centering
    \subfloat{\definecolor{mycolor1}{rgb}{0.00000,0.44700,0.74100}%
\begin{tikzpicture}

\begin{axis}[%
width=0.3\columnwidth,
height=0.186\columnwidth,
at={(0\columnwidth,0\columnwidth)},
scale only axis,
bar shift auto,
xticklabels={$\|\Delta_{\text{IPM}}\|_F$,$\|G-\hat{G}_{\text{IPM}}\|_F$,$\|\Delta_{\text{GU}}\|_F$,$\|G-\hat{G}_{GU}\|_F$,$\|\Delta_{\text{TSDP}(G)}\|_F$,$\|\Delta_{\text{TSDP}(A+I)}\|_F$},
enlarge x limits=0.2,
xtick=data,
xticklabel style={rotate=90,font=\footnotesize},
ymin=0,
ymax=7,
axis background/.style={fill=white}
]
\addplot[ybar, fill=mycolor1, draw=black, area legend] table[row sep=crcr] {%
0 6.70345808700298\\
1 1.68129422639828\\
2 5.15634513007784\\
3 1.3060423235387\\
4 0.914487708990402\\
5 1.05550236480087\\
};
\end{axis}
\end{tikzpicture}%
}
    \subfloat{\definecolor{mycolor1}{rgb}{0.00000,0.44700,0.74100}%
\begin{tikzpicture}

\begin{axis}[%
width=0.204\columnwidth,
height=0.204\columnwidth,
at={(0\columnwidth,0\columnwidth)},
scale only axis,
xmin=0,
xmax=35,
xlabel style={font=\color{white!15!black}},
xlabel={nz = 35},
y dir=reverse,
ymin=0,
ymax=35,
axis background/.style={fill=white},
axis x line*=bottom,
axis y line*=left
]
\addplot [color=mycolor1, line width=2.0pt, only marks, mark size=1pt, mark=*, mark options={solid, mycolor1}, forget plot]
  table[row sep=crcr]{%
1	5\\
1	8\\
1	11\\
1	12\\
1	13\\
1	18\\
1	20\\
1	22\\
2	1\\
2	3\\
2	4\\
2	8\\
2	9\\
2	14\\
2	18\\
2	20\\
2	22\\
2	31\\
3	1\\
3	2\\
3	4\\
3	8\\
3	9\\
3	10\\
3	14\\
3	28\\
3	29\\
3	33\\
4	1\\
4	2\\
4	3\\
4	8\\
4	13\\
4	14\\
4	25\\
5	1\\
5	6\\
5	7\\
5	11\\
6	1\\
6	7\\
6	9\\
6	11\\
6	17\\
7	1\\
7	5\\
7	6\\
7	17\\
8	1\\
8	2\\
8	3\\
8	4\\
8	14\\
9	1\\
9	3\\
9	5\\
9	30\\
9	31\\
9	33\\
9	34\\
10	3\\
10	7\\
10	20\\
10	34\\
11	1\\
11	5\\
11	6\\
11	14\\
11	30\\
12	1\\
12	7\\
12	9\\
13	1\\
13	4\\
13	24\\
14	1\\
14	2\\
14	3\\
14	4\\
14	7\\
14	24\\
14	34\\
15	32\\
15	33\\
15	34\\
16	3\\
16	33\\
16	34\\
17	6\\
17	7\\
18	1\\
18	2\\
19	14\\
19	33\\
19	34\\
20	1\\
20	2\\
20	32\\
20	34\\
21	33\\
21	34\\
22	1\\
22	2\\
22	31\\
22	33\\
23	3\\
23	26\\
23	28\\
23	33\\
23	34\\
24	26\\
24	28\\
24	30\\
24	33\\
24	34\\
25	6\\
25	26\\
25	28\\
25	32\\
26	11\\
26	24\\
26	25\\
26	31\\
26	32\\
27	1\\
27	20\\
27	30\\
27	34\\
28	3\\
28	6\\
28	24\\
28	25\\
28	32\\
28	34\\
29	3\\
29	7\\
29	32\\
29	34\\
30	4\\
30	24\\
30	27\\
30	33\\
30	34\\
31	2\\
31	4\\
31	9\\
31	24\\
31	33\\
31	34\\
32	1\\
32	25\\
32	26\\
32	28\\
32	29\\
32	33\\
32	34\\
33	3\\
33	9\\
33	15\\
33	16\\
33	19\\
33	21\\
33	23\\
33	24\\
33	30\\
33	31\\
33	32\\
33	34\\
34	10\\
34	15\\
34	16\\
34	19\\
34	20\\
34	21\\
34	23\\
34	27\\
34	29\\
};
\addplot [color=red, line width=2.0pt, only marks, mark=x, mark options={solid, red}, forget plot]
  table[row sep=crcr]{%
1	1\\
1	2\\
1	3\\
1	4\\
1	6\\
1	7\\
1	9\\
1	14\\
1	24\\
1	28\\
1	30\\
1	31\\
1	32\\
1	33\\
1	34\\
34	1\\
34	2\\
34	3\\
34	4\\
34	5\\
34	6\\
34	7\\
34	8\\
34	9\\
34	11\\
34	14\\
34	24\\
34	25\\
34	26\\
34	28\\
34	30\\
34	31\\
34	32\\
34	33\\
34	34\\
};
\end{axis}

\begin{axis}[%
width=0.204\columnwidth,
height=0.204\columnwidth,
at={(0.279\columnwidth,0\columnwidth)},
scale only axis,
xmin=0,
xmax=35,
xlabel style={font=\color{white!15!black}},
xlabel={nz = 24},
y dir=reverse,
ymin=0,
ymax=35,
axis background/.style={fill=white},
axis x line*=bottom,
axis y line*=left
]
\addplot [color=mycolor1, line width=2.0pt, only marks, mark size=1pt, mark=*, mark options={solid, mycolor1}, forget plot]
  table[row sep=crcr]{%
1	5\\
1	8\\
1	11\\
1	12\\
1	13\\
1	18\\
1	20\\
1	22\\
2	1\\
2	3\\
2	4\\
2	8\\
2	14\\
2	18\\
2	20\\
2	22\\
2	31\\
3	1\\
3	2\\
3	4\\
3	8\\
3	9\\
3	10\\
3	14\\
3	28\\
3	29\\
3	33\\
4	1\\
4	3\\
4	8\\
4	13\\
4	14\\
5	1\\
5	7\\
5	11\\
6	1\\
6	7\\
6	11\\
6	17\\
7	1\\
7	5\\
7	6\\
7	17\\
8	1\\
8	2\\
8	3\\
8	4\\
9	1\\
9	3\\
9	31\\
9	33\\
9	34\\
10	3\\
10	34\\
11	1\\
11	5\\
11	6\\
12	1\\
13	1\\
13	4\\
14	1\\
14	2\\
14	3\\
14	4\\
14	34\\
15	33\\
15	34\\
16	33\\
16	34\\
17	6\\
17	7\\
18	1\\
18	2\\
19	33\\
19	34\\
20	1\\
20	2\\
20	34\\
21	33\\
21	34\\
22	1\\
22	2\\
23	33\\
23	34\\
24	26\\
24	28\\
24	30\\
24	33\\
24	34\\
25	26\\
25	28\\
25	32\\
26	24\\
26	25\\
26	32\\
27	30\\
27	34\\
28	3\\
28	24\\
28	25\\
28	34\\
29	3\\
29	32\\
29	34\\
30	24\\
30	27\\
30	33\\
30	34\\
31	2\\
31	9\\
31	34\\
32	1\\
32	25\\
32	26\\
32	29\\
32	33\\
32	34\\
33	3\\
33	9\\
33	15\\
33	16\\
33	19\\
33	21\\
33	23\\
33	24\\
33	30\\
33	31\\
33	32\\
33	34\\
34	10\\
34	15\\
34	16\\
34	19\\
34	20\\
34	21\\
34	23\\
34	27\\
34	29\\
};
\addplot [color=red, line width=2.0pt, only marks, mark=x, mark options={solid, red}, forget plot]
  table[row sep=crcr]{%
1	1\\
1	2\\
1	3\\
1	4\\
1	6\\
1	7\\
1	9\\
1	14\\
1	32\\
4	2\\
25	1\\
26	1\\
29	1\\
31	33\\
33	1\\
34	1\\
34	9\\
34	14\\
34	24\\
34	28\\
34	30\\
34	31\\
34	32\\
34	33\\
};
\end{axis}

\begin{axis}[%
width=0.633\columnwidth,
height=0.258\columnwidth,
at={(-0.086\columnwidth,-0.035\columnwidth)},
scale only axis,
xmin=0,
xmax=1,
ymin=0,
ymax=1,
axis line style={draw=none},
ticks=none,
axis x line*=bottom,
axis y line*=left
]
\end{axis}
\end{tikzpicture}%
}
    
    \caption{Enforcing PageRank. Norm of the perturbation (left panel) and lack of sign preservation for the underlying sparse matrix for the TSDP perturbations (right panels).}
    \label{fig:pr_norms_and_sign_bug}
\end{figure}
Regarding the norms of the perturbations (Fig.~\ref{fig:pr_norms_and_sign_bug}), our method, which acts only on the adjacency matrix, results in a significantly larger perturbation. This perturbation is considerably reduced when applied to the entire stochastic matrix. Furthermore, since the method implemented in TSDP uses the commercial solver GUROBI for optimization, for completeness in the tests we also report our algorithm in which we replace the IPM described in Section~\ref{sec:ipm-solve} with the appropriate GUROBI routine. We observe that the latter implementation reaches a smaller perturbation {at the price of producing a higher number of nonzero entries, cfr., Fig.~\ref{fig:pr_pattern} and Fig.~\ref{fig:pr_norms_and_sign_bug}. This might be probably attributed to the fact that, due to the non-uniqueness of the solution, the two solvers converge to different solutions: one that promotes the optimization of the perturbation norm (GUROBI) and one that promotes the optimization of the $1$-norm (IPM).} Importantly, the perturbation induced on the adjacency matrix by the TSDP method, $G + \Delta_{\text{TDSP}} - (1-\alpha)\mathbf{1}\mathbf{t}^{\top}$, when applied to the dense stochastic matrix $G$, {with either the pattern of $G$ or the pattern of $I+A$,} produces unacceptable negative edge weights, as indicated by the red crosses in Fig.~\ref{fig:pr_norms_and_sign_bug}.

\subsubsection{S1 and S2 Scenarios} \label{sec:PR_S1_S2}
In Tables~\ref{tab:S1-PR_tab1} and~\ref{tab:S1-PR_tab2} we report the results for the PageRank problem for the S1 scenario whereas in Tables~\ref{tab:S2-PR_tab1} and~\ref{tab:S2-PR_tab2} we report the same results for the Scenario S2.
The presented numerical results, show, again, the effectiveness of our proposal achieving, for most of the problems Kendall's correlation coefficient $\kappa_\tau$  very close or equal to~$1$. The worst situations, with $\kappa_\tau$
in the range $[0.86,0.87]$
occur when the conditioning of $LL^{\top}$ is greater than $10^{35}$; however, these values of $\kappa_\tau$ are still  very good results in terms of ordering of the computed scores. {And indeed, to the best of our knowledge, the main reason for such mismatched orderings is mainly due to a fatal interaction between the finite precision and the inherent ill-conditioning of the problems. The numerical results presented in Section \ref{sec:augmented_precision} confirm that, when extended precision is used, the error between the computed target distribution and the actual one decreases proportionally to the tolerance used in the optimization solver.}

The effectiveness of the proposed method is also measured in terms of the  coefficient $\widehat{r}$ (see Proposition \ref{pro:interpretability}), which lies in the interval $[1,2.51]$, confirming that the obtained target centrality can be interpreted as the PageRank of a modified network having teleportation parameter $\widehat{\alpha}$ very close to the teleportation parameter $\alpha$ of the original PageRank problem, see Remark \ref{rema:r_hat} for more details. The comparison of Tables~\ref{tab:S1-PR_tab1} and~\ref{tab:S1-PR_tab2}  and Tables~\ref{tab:S2-PR_tab1} and~\ref{tab:S2-PR_tab2} show counter-intuitive occurrences for the relative norm of the produced perturbations $\Delta$: for a relatively large set of problems, see, e.g., problems \texttt{6], 9], 13], 15]} for S1 and problems \texttt{6], 13]} for S2, the norm of the obtained perturbation for the sparsified version is smaller than the norm obtained in the pure Frobenius norm minimization.

Finally, comparing the results presented in this section with those presented in Sections~\ref{sec:Katz_S1} and~\ref{sec:Katz_S2}, it is interesting to note how assigning the Pagerank centrality to a given network is a substantially more challenging computational problem when compared to Katz centrality assignment problem. This is mainly due to the conditioning of the involved constraint matrices and the consistently denser Cholesky factors of \eqref{eq:schur_IPM} arising in this case. 

\begin{table}[htbp!]\centering
\caption{Pagerank. S1 Scenario. $\beta=1$. The numbering of the test cases corresponds to the matrices in Table~\ref{tab:dataset}. The table reports:  the conditioning of the constraints matrix, the number of IPM iterations,
the elapsed time measured in seconds, the relative norm of the obtained perturbation together with the number of nonzero entries of positive ($+$), negative $(-)$ sign and the number of non-zero elements in the Cholesky factor of equation \eqref{eq:schur_IPM} in the last IPM iteration. The correlation between the obtained centrality measure with the target one is measured in terms of Kendall's $\kappa_\tau$.  In the last column, we report $\widehat{r}$ from Proposition \ref{pro:interpretability}; ``$^*$'' is reported when the matrix $A+\Delta$ has nonnegative diagonal elements and the shift in Proposition \ref{pro:interpretability} is not needed, see also Remark \ref{rema:r_hat}. }\label{tab:S1-PR_tab1}
\begin{tabular}{rllllrrrccc}
\toprule\small
& $\operatorname{cond}(LL^{\top})$ & Iter & T (s) & $\nicefrac{\|\Delta\|_F}{\|A\|_F}$ & $+$  & $-$ & \texttt{nnz Chol.} & $\kappa_\tau$ & $\widehat{r}$ \\
&     &      &       &                &      &     & of \eqref{eq:schur_IPM} & & \\

\midrule
1 ]  & 5.929e+17                                         & 7        & 11.30       & 1.739e-05      & 145810        & 156273          & 20890924     & 1.00 &     1.00  \\    
 2 ]  & 2.236e+20                                         & 19        & 1.53       & 4.640e-01      & 28199        & 31112          & 271830     & 1.00 &     1.62  \\    
 3 ]  & 1.384e+19                                         & 8        & 13.24       & 7.767e-03      & 69251        & 81393          & 22380215     & 1.00 &     $^*$  \\    
 4 ]  & 5.285e+18                                         & 9        & 153.72       & 8.516e-03      & 259045        & 300677          & 220271238     & 1.00 &     $^*$  \\    
 5 ]  & 7.658e+18                                         & 6        & 3.98       & 3.016e-03      & 54770        & 55282          & 10292515     & 1.00 &     1.01  \\    
 6 ]  & 1.486e+38                                         & 21        & 13.32       & 1.130e-07      & 211926        & 191241          & 2391378     & 0.99 &     1.00  \\    
 7 ]  & 3.154e+18                                         & 6        & 3.85       & 4.676e-03      & 56301        & 56754          & 8337719     & 1.00 &     1.01  \\    
 8 ]  & 4.229e+18                                         & 7        & 0.32       & 2.231e-02      & 16363        & 16674          & 218301     & 1.00 &     1.01  \\    
 9 ]  & 2.071e+37                                         & 22        & 3.82       & 1.312e-06      & 73942        & 66212          & 731975     & 0.99 &     1.00  \\    
 10 ]  & 2.282e+19                                         & 6        & 6.37       & 3.415e-02      & 218455        & 240231          & 7742048     & 1.00 &     1.02  \\    
 11 ]  & 3.389e+18                                         & 7        & 0.98       & 3.688e-02      & 37603        & 39158          & 951062     & 1.00 &     1.02  \\    
 12 ]  & 3.193e+18                                         & 9        & 0.13       & 5.869e-02      & 2576        & 3087          & 127390     & 1.00 &     1.11  \\    
 13 ]  & 5.606e+37                                         & 22        & 9.34       & 2.098e-07      & 151151        & 132152          & 1473042     & 0.99 &     1.00  \\    
 14 ]  & 3.741e+18                                         & 6        & 0.14       & 2.352e-02      & 9090        & 9219          & 104772     & 1.00 &     1.07  \\    
 15 ]  & 3.572e+37                                         & 22        & 5.77       & 2.713e-07      & 99620        & 88483          & 995357     & 0.99 &     1.00  \\    

\bottomrule
\end{tabular}
\end{table}

\begin{table}[htbp!]\centering
\caption{Pagerank. S1 Scenario. $(1-\beta)/\beta=10$. The numbering of the test cases corresponds to the matrices in Table~\ref{tab:dataset}. The table reports:  the conditioning of the constraints matrix, the number of IPM iterations,
the elapsed time measured in seconds, the relative norm of the obtained perturbation together with the number of nonzero entries of positive ($+$), negative $(-)$ sign and the number of non-zero elements in the Cholesky factor of equation \eqref{eq:schur_IPM} in the last IPM iteration. The correlation between the obtained
centrality measure with the target one is measured in terms of Kendall's $\kappa_\tau$.  In the last column, we report $\widehat{r}$ from Proposition \ref{pro:interpretability}.}\label{tab:S1-PR_tab2}
\begin{tabular}{rllllrrrccc}
\toprule
& $\operatorname{cond}(LL^{\top})$ & Iter & T (s) & $\nicefrac{\|\Delta\|_F}{\|A\|_F}$ & $+$  & $-$ & \texttt{nnz Chol.} & $\kappa_\tau$ & $\widehat{r}$ \\
&     &      &       &                &      &     & of \eqref{eq:schur_IPM} & & \\

\midrule
  1 ]  & 3.133e+15                                       & 16       & 45.24    & 1.408e-06      & 65371        & 65263          & 21957743  &   1.00  &    1.00 \\   
 2 ]  & 1.136e+20                                       & 27       & 6.03    & 8.334e-01      & 5180        & 9901          & 412477  &   0.99  &    1.70 \\   
 3 ]  & 2.567e+17                                       & 23       & 42.58    & 2.254e-02      & 2837        & 4014          & 23352601  &   1.00  &    1.09 \\   
 4 ]  & 5.618e+18                                       & 29       & 520.11    & 2.347e-02      & 9378        & 14552          & 222970119  &   1.00  &    1.05 \\   
 5 ]  & 7.202e+16                                       & 30       & 28.11    & 4.342e-03      & 6485        & 6998          & 10553033  &   1.00  &    1.01 \\   
 6 ]  & 7.579e+37                                       & 20       & 34.53    & 1.266e-16      & 104        & 86          & 3128206  &   0.99  &    1.00 \\   
 7 ]  & 3.814e+16                                       & 28       & 24.65    & 7.915e-03      & 5100        & 5466          & 9514456  &   1.00  &    1.01 \\   
 8 ]  & 8.953e+16                                       & 24       & 2.49    & 4.427e-02      & 1458        & 1709          & 309461  &   1.00  &    1.04 \\   
 9 ]  & 7.029e+35                                       & 19       & 7.69    & 1.911e-16      & 99        & 122          & 1049107  &   0.99  &    1.00 \\   
 10 ]  & 1.822e+17                                       & 32       & 72.17    & 6.627e-02      & 31375        & 49225          & 8833162  &   1.00  &    1.13 \\   
 11 ]  & 1.067e+17                                       & 27       & 7.53    & 7.329e-02      & 5744        & 7829          & 1164779  &   1.00  &    1.04 \\   
 12 ]  & 5.441e+17                                       & 15       & 0.40    & 9.276e-02      & 217        & 325          & 131183  &   1.00  &    1.27 \\   
 13 ]  & 1.501e+37                                       & 20       & 20.63    & 1.237e-16      & 179        & 232          & 2020942  &   0.99  &    1.00 \\   
 14 ]  & 8.706e+16                                       & 23       & 1.21    & 3.899e-02      & 2327        & 2417          & 145116  &   1.00  &    1.07 \\   
 15 ]  & 1.005e+38                                       & 20       & 9.66    & 1.278e-16      & 158        & 416          & 1382293  &   0.99  &    1.00 \\

\bottomrule
\end{tabular}
\end{table}

\begin{table}[htbp]\centering
\caption{Pagerank. S2 Scenario. $\beta=1$. The numbering of the test cases corresponds to the matrices in Table~\ref{tab:dataset}. The table reports:  the conditioning of the constraints matrix, the number of IPM iterations,
the elapsed time measured in seconds, the relative norm of the obtained perturbation together with the number of nonzero entries of positive ($+$), negative $(-)$ sign and the number of non-zero elements in the Cholesky factor of equation \eqref{eq:schur_IPM} in the last IPM iteration. The correlation between the obtained
centrality measure with the target one is measured in terms of Kendall's $\kappa_\tau$. In the last column, we report $\widehat{r}$ from Proposition \ref{pro:interpretability}; ``$^*$'' is reported when the matrix $A+\Delta$ has nonnegative diagonal elements and the shift in Proposition \ref{pro:interpretability} is not needed, see also Remark \ref{rema:r_hat}.}\label{tab:S2-PR_tab1}
\begin{tabular}{rllllrrrccc}
\toprule
& $\operatorname{cond}(LL^{\top})$ & Iter & T (s) & $\nicefrac{\|\Delta\|_F}{\|A\|_F}$ & $+$  & $-$ & \texttt{nnz Chol.} & $\kappa_\tau$ & $\widehat{r}$ \\
&     &      &       &                &      &     & of \eqref{eq:schur_IPM} & & \\

\midrule
1 ]  & 5.651e+19                                        & 7           & 11.12   & 1.459e-02      & 167598        & 134627          & 20890924   &  1.00  &  1.00 \\    
 2 ]  & 3.368e+21                                        & 18           & 1.44   & 4.954e-01      & 36909        & 22401          & 271830   &  1.00  &  2.51 \\    
 3 ]  & 1.735e+20                                        & 10           & 16.36   & 1.698e-02      & 71556        & 79079          & 22380215   &  1.00  &  $^*$ \\    
 4 ]  & 8.941e+18                                        & 11           & 183.88   & 1.834e-02      & 273699        & 286016          & 220271238   &  1.00  &  $^*$ \\    
 5 ]  & 7.418e+18                                        & 6           & 4.05   & 4.364e-02      & 55110        & 55105          & 10292515   &  1.00  &  1.06 \\    
 6 ]  & 1.856e+37                                        & 21           & 12.98   & 6.638e-07      & 204557        & 180803          & 2391523   &  0.87  &  1.00 \\    
 7 ]  & 4.691e+18                                        & 7           & 4.09   & 1.117e-02      & 57142        & 56150          & 8337719   &  1.00  &  1.01 \\    
 8 ]  & 2.801e+18                                        & 6           & 0.30   & 5.065e-02      & 16776        & 16156          & 218301   &  1.00  &  1.02 \\    
 9 ]  & 1.136e+36                                        & 22           & 3.98   & 4.170e-06      & 72758        & 67240          & 731975   &  0.87  &  1.00 \\    
 10 ]  & 1.862e+21                                        & 8           & 8.79   & 9.936e-02      & 234054        & 224630          & 7742048   &  1.00  &  1.05 \\    
 11 ]  & 8.015e+19                                        & 7           & 0.84   & 8.065e-02      & 39085        & 37694          & 951062   &  1.00  &  1.08 \\    
 12 ]  & 8.825e+18                                        & 9           & 0.14   & 1.631e-01      & 2816        & 2847          & 127390   &  1.00  &  1.22 \\    
 13 ]  & 1.515e+37                                        & 22           & 9.03   & 1.650e-06      & 146765        & 133209          & 1472809   &  0.87  &  1.00 \\    
 14 ]  & 8.591e+18                                        & 6           & 0.14   & 3.612e-02      & 8523        & 8227          & 104772   &  1.00  &  1.12 \\    
 15 ]  & 9.993e+36                                        & 21           & 5.31   & 2.581e-06      & 94505        & 91419          & 995261   &  0.86  &  1.00 \\    

\bottomrule
\end{tabular}
\end{table}

\begin{table}[htbp]\centering
\caption{Pagerank. S2 Scenario. $(1-\beta)/\beta=100$. The numbering of the test cases corresponds to the matrices in Table~\ref{tab:dataset}. The table reports:  the conditioning of the constraints matrix, the number of IPM iterations,
the elapsed time measured in seconds, the relative norm of the obtained perturbation together with the number of nonzero entries of positive ($+$), negative $(-)$ sign and the number of non-zero elements in the Cholesky factor of equation \eqref{eq:schur_IPM} in the last IPM iteration. The correlation between the obtained
centrality measure with the target one is measured in terms of Kendall's $\kappa_\tau$.  In the last column, we report $\widehat{r}$ from Proposition \ref{pro:interpretability}; ``$^*$'' is reported when the matrix $A+\Delta$ has nonnegative diagonal elements and the shift in Proposition \ref{pro:interpretability} is not needed, see also Remark \ref{rema:r_hat}.}\label{tab:S2-PR_tab2}
\begin{tabular}{rllllrrrccc}
\toprule
& $\operatorname{cond}(LL^{\top})$ & Iter & T (s) & $\nicefrac{\|\Delta\|_F}{\|A\|_F}$ & $+$  & $-$ & \texttt{nnz Chol.} & $\kappa_\tau$ & $\widehat{r}$ \\
&     &      &       &                &      &     & of \eqref{eq:schur_IPM} & & \\

\midrule
1 ]  & 3.793e+15                                        & 40        & 111.37   & 6.573e-02      & 52098        & 51119          & 21957743    &  1.00    &  1.00 \\    
 2 ]  & 3.291e+21                                        & 29        & 6.44   & 6.486e-01      & 4176        & 5469          & 412477    &  1.00    &  2.51 \\    
 3 ]  & 6.681e+17                                        & 31        & 56.13   & 3.227e-02      & 1974        & 2261          & 23352601    &  1.00    &  $^*$ \\    
 4 ]  & 5.070e+18                                        & 33        & 594.38   & 3.626e-02      & 8814        & 10254          & 222970119    &  1.00    &  $^*$ \\    
 5 ]  & 1.036e+17                                        & 30        & 28.99   & 6.616e-02      & 7902        & 9589          & 10553033    &  1.00    &  1.07 \\    
 6 ]  & 1.188e+39                                        & 20        & 34.38   & 1.235e-16      & 77        & 83          & 3127868    &  0.87    &  1.00 \\    
 7 ]  & 4.017e+16                                        & 28        & 23.58   & 2.037e-02      & 4306        & 4591          & 9514456    &  1.00    &  1.01 \\    
 8 ]  & 7.890e+16                                        & 14        & 1.56   & 1.163e-01      & 189        & 708          & 309461    &  1.00    &  1.16 \\    
 9 ]  & 6.041e+35                                        & 27        & 11.46   & 1.514e-06      & 48        & 40          & 1044053    &  0.87    &  1.00 \\    
 10 ]  & 4.907e+17                                        & 37        & 88.26   & 1.613e-01      & 25733        & 34290          & 8833162    &  1.00    &  1.37 \\    
 11 ]  & 2.801e+17                                        & 44        & 12.50   & 1.429e-01      & 8067        & 7876          & 1164779    &  1.00    &  1.28 \\    
 12 ]  & 1.753e+18                                        & 14        & 0.37   & 2.572e-01      & 226        & 261          & 131183    &  1.00    &  1.98 \\    
 13 ]  & 1.823e+37                                        & 20        & 22.30   & 2.322e-16      & 201        & 188          & 2020506    &  0.87    &  1.00 \\    
 14 ]  & 8.752e+16                                        & 21        & 1.05   & 4.908e-02      & 744        & 842          & 145116    &  1.00    &  1.16 \\    
 15 ]  & 8.733e+35                                        & 37        & 21.46   & 2.086e-06      & 201        & 251          & 1378450    &  0.86    &  1.00 \\    

\bottomrule
\end{tabular}
\end{table}

\section{Conclusions}
\label{sec:conclusions}

In this paper, we addressed the problem of determining the minimum norm perturbation required to enforce desired centrality measures, specifically Katz and PageRank centralities, in complex networks. Our approach utilizes scalable optimization procedures to compute these perturbations, allowing for targeted modifications of the network's structure. Through numerical experiments on real-world networks, we demonstrated the effectiveness of our method in achieving specific centrality objectives while preserving network sparsity structure. Future work will explore extending this framework to other centrality measures and further improving the computational efficiency of our proposal {using specialised tools, e.g., specialised preconditioned iterative methods, for the solution of linear systems arising from the Interior Point Method.}

\appendix

\section*{Acknowledgments}
We thank the two anonymous referees for their valuable comments and suggestions, which have improved the quality of this paper.

\bibliographystyle{siamplain}
\bibliography{references}

\cleardoublepage

\section{Kendall's  correlation coefficient} \label{sec:kendal}

Kendall's $\kappa_\tau$ is a non-parametric statistic commonly used to compare two rankings by quantifying the degree of agreement between them. Suppose we have two rankings, $r_1$ and $r_2$, of the same set of $n$ items. For any pair of distinct items $(i,j)$, if the order is the same in both rankings (i.e., either $r_1(i) < r_1(j)$ and $r_2(i) < r_2(j)$, or $r_1(i) > r_1(j)$ and $r_2(i) > r_2(j)$), the pair is considered \emph{concordant}; if the order differs, the pair is \emph{discordant}. Kendall's $\kappa_\tau$ is then defined as
\[
\kappa_\tau = \frac{n_c - n_d}{\binom{n}{2}},
\]
where $n_c$ and $n_d$ denote the number of concordant and discordant pairs, respectively, and $\binom{n}{2}$ is the total number of pairs. The coefficient ranges from $-1$ (complete disagreement) to $1$ (complete agreement), with $0$ indicating no association between the rankings. This measure is especially useful when the focus is on the ordinal relationship between items, making it robust in contexts where only the order, rather than the magnitude, is of interest. 

\begin{example}
	To have an example we may consider two orderings:
	\begin{align*}  
		\mathbf{r}_1 = &\; \begin{bmatrix}    1 &     2 &    3 &    4  &   5  &   6 &    7 &    8  &   9 &   10 \end{bmatrix}^\top, \\
		\mathbf{r}_2 =&\;  \begin{bmatrix}   1  &   3 &    2   &  4 &    6  &   5 &    8  &   7  &   9  &  10 \end{bmatrix}^\top,
	\end{align*}
	In this case, after computing the number of concordant and discordant pairs, one obtains 
	\[
	\kappa_\tau = \text{\lstinline[style=Matlab-editor]{corr(r1,r2,'Type','Kendall')}} \approx 0.8667,
	\]
	which indicates a high degree of agreement between the two orderings. It is important to note that Kendall's tau treats every mismatch equally, regardless of its position in the list. Whether the disagreement occurs among the top-ranked elements or the bottom-ranked ones, the contribution to the overall coefficient remains the same. This means that the metric does not capture the relative importance of errors based on their positions in the list.
\end{example}

\appendixnotitle{Additional numerical experiments} 
{In this section, we present additional experiments that complement those discussed in the paper. Specifically, in Table~\ref{tab:gurobi_vs_ipm}, we report the results of experiments conducted on instances with the most ill-conditioned constraint matrix $LL^\top$ when solved using the commercial GUROBI solver\footnote{See \href{https://www.gurobi.com/}{www.gurobi.com}.} and the Interior Point Method based on~\cite{MR4594481,CIPOLLA2023}.  The comparison of these results highlights that, regardless of whether the problems are strictly or weakly convex, the optimizer implemented in GUROBI---which does not exploit proximal regularization as in PS-IPM---fails to progress in optimizing the given objective function, as evidenced by comparing the obtained objective values in columns $\nicefrac{\| \Delta \|_F}{\| A \|_F}$ and $\| \Delta \|_F$.  
}

\begin{table}[htbp]
	\centering
	\setlength{\tabcolsep}{0.35em}
	\small
	
	\begin{tabular}{llcccccccc}
		\toprule
		& & \multicolumn{7}{c}{IPM} \\
		\midrule
		Matrix & $\beta$ & Time & $\kappa_\tau$ & $\hat{r}$ & $\| \Delta \|_F$ & $\nicefrac{\| \Delta \|_F}{\| A \|_F}$ nnz & Iter. \\
		\toprule
		PGPgiantcompo & 1.0000 & 2.11e+00 & 1.00 & 1.62 & 1.02e+02 & 4.64e-01 & 59312 & 19 \\
		& 0.5000 & 6.16e+00 & 1.00 & 1.53 & 1.06e+02 & 4.79e-01 & 52980 & 24 \\
		& 0.0099 & 7.04e+00 & 0.99 & 1.70 & 1.84e+02 & 8.33e-01 & 49969 & 27 \\
		& 0.0050 & 7.33e+00 & 0.99 & 1.86 & 2.13e+02 & 9.64e-01 & 51854 & 29 \\
		\toprule
		ct2010 & 1.0000 & 1.44e+01 & 0.99 & 1.00 & 1.20e-02 & 3.66e-10 & 399807 & 21 \\
		& 0.5000 & 3.72e+01 & 0.99 & 1.00 & 1.34e-07 & 4.10e-15 & 134019 & 22 \\
		& 0.0099 & 3.27e+01 & 0.99 & 1.00 & 4.09e-09 & 1.25e-16 & 77908 & 20 \\
		& 0.0050 & 3.27e+01 & 0.99 & 1.00 & 4.13e-09 & 1.26e-16 & 68009 & 20 \\
		\toprule
		nh2010 & 1.0000 & 8.74e+00 & 0.99 & 1.00 & 1.05e-02 & 2.46e-10 & 282566 & 22 \\
		& 0.5000 & 2.26e+01 & 0.99 & 1.00 & 1.20e-05 & 2.81e-13 & 96752 & 21 \\
		& 0.0099 & 2.10e+01 & 0.99 & 1.00 & 5.27e-09 & 1.24e-16 & 49659 & 20 \\
		& 0.0050 & 1.88e+01 & 0.99 & 1.00 & 5.75e-09 & 1.35e-16 & 49899 & 19 \\
		\toprule
		vt2010 & 1.0000 & 6.24e+00 & 0.99 & 1.00 & 1.23e-02 & 2.90e-10 & 187451 & 22 \\
		& 0.5000 & 1.40e+01 & 0.99 & 1.00 & 4.83e-06 & 1.13e-13 & 40800 & 26 \\
		& 0.0099 & 1.05e+01 & 0.99 & 1.00 & 5.53e-09 & 1.30e-16 & 34743 & 20 \\
		& 0.0050 & 1.05e+01 & 0.99 & 1.00 & 5.94e-09 & 1.39e-16 & 37013 & 20 \\
		\bottomrule
	\end{tabular}
	\medskip
	
	\begin{tabular}{llcccccccc}
		\toprule
		& & \multicolumn{7}{c}{GUROBI} \\
		\midrule
		Matrix & $\beta$ & Time & $\kappa_\tau$ & $\hat{r}$   & $\| \Delta \|_F$ & $\nicefrac{\| \Delta \|_F}{\| A \|_F}$  & nnz & Iter. \\
		\midrule
		PGPgiantcompo & 1.0000 & 4.49e-01 & 1.00 & 1.62 & 1.02e+02 & 4.64e-01 & 59312 & 17 \\
		& 0.5000 & 1.04e+00 & 1.00 & 1.47 & 1.06e+02 & 4.79e-01 & 59312 & 21 \\
		& 0.0099 & 1.17e+00 & 1.00 & 1.61 & 1.24e+02 & 5.63e-01 & 59312 & 26 \\
		& 0.0050 & 1.26e+00 & 1.00 & 1.70 & 1.29e+02 & 5.83e-01 & 59311 & 27 \\
		\midrule
		ct2010 & 1.0000 & 2.93e+00 & 1.00 & 3.59 & 6.49e+06 & 1.98e-01 & 403930 & 20 \\
		& 0.5000 & 6.64e+00 & 1.00 & 4.95 & 1.23e+07 & 3.75e-01 & 403930 & 17 \\
		& 0.0099 & 6.14e+00 & 1.00 & 4.95 & 1.23e+07 & 3.75e-01 & 403930 & 16 \\
		& 0.0050 & 6.30e+00 & 1.00 & 4.94 & 1.23e+07 & 3.75e-01 & 403927 & 17 \\
		\midrule
		nh2010 & 1.0000 & 2.22e+00 & 1.00 & 6.05 & 1.40e+07 & 3.29e-01 & 283387 & 21 \\
		& 0.5000 & 4.92e+00 & 1.00 & 7.45 & 1.95e+07 & 4.59e-01 & 283387 & 20 \\
		& 0.0099 & 4.57e+00 & 1.00 & 7.45 & 1.95e+07 & 4.59e-01 & 283387 & 19 \\
		& 0.0050 & 4.45e+00 & 1.00 & 7.45 & 1.95e+07 & 4.59e-01 & 283386 & 19 \\
		\midrule
		vt2010 & 1.0000 & 1.20e+00 & 1.00 & 5.20 & 9.77e+06 & 2.29e-01 & 188178 & 21 \\
		& 0.5000 & 2.92e+00 & 1.00 & 7.33 & 1.69e+07 & 3.96e-01 & 188178 & 18 \\
		& 0.0099 & 2.80e+00 & 1.00 & 7.33 & 1.69e+07 & 3.96e-01 & 188178 & 17 \\
		& 0.0050 & 2.63e+00 & 1.00 & 7.32 & 1.68e+07 & 3.96e-01 & 188176 & 17 \\
		\bottomrule
	\end{tabular}
	
	\caption{Comparison of the solution obtained for solving the PageRank modification problem from Section~\ref{sec:pagerank}.}
	\label{tab:gurobi_vs_ipm}
\end{table}

\section{Experiments with augmented precision} \label{sec:augmented_precision}
To complete this overview of the optimizer, we also consider the case where the precision is pushed beyond double precision. For this, we utilize the ADVANPIX Multiprecision Computing Toolbox for MATLAB\footnote{See \href{https://www.advanpix.com/}{www.advanpix.com}.}, \texttt{v.4.8.0}, which extends MATLAB's capabilities to support arbitrary precision arithmetic.

Specifically, the number of digits used in computations can be set via \lstinline[style=Matlab-editor]{mp.Digits(71);}, which corresponds to the octuple-precision floating-point format (1 bit for the sign, 19 bits for the exponent, and 236 explicitly stored bits for the significand).

Once this precision is set, various MATLAB functions, such as eigenvalue computation, solution of linear systems, and arithmetic operations, are automatically overloaded to use the new arithmetic for variables of class \lstinline[style=Matlab-editor]{mp}. 

With this structure we can then try to push the optimization algorithm to obtain a higher precision on the obtained ranking vector for both the Katz (Section~\ref{sec:katz}) and PageRank (Section~\ref{sec:pagerank}) problems. In Figure~\ref{fig:katz} we start from an example with the Katz ranking optimization problem for $\beta=1$---i.e., no $\|\cdot\|_1$-penalty term---and $\beta = 0.5$---i.e., convex combination of the $\|\cdot\|_2$ and $\|\cdot\|_1$-penalties. We use the \texttt{Newman/karate} network. This is a small network with 34 nodes and 156 edges. We select a $\widehat{\boldsymbol{\mu}}$ target vector equal to the Katz centrality vector $\boldsymbol{\mu}$ obtained for $\alpha = \nicefrac{\rho(A)}{2}$ in all the entries except for $\widehat{\boldsymbol{\mu}}_{33} = \nicefrac{\boldsymbol{\mu}_{34}}{2}$ and $\widehat{\boldsymbol{\mu}}_{5} = 1.5\times\boldsymbol{\mu}_{5}$.  
\begin{figure}[htbp]
	\centering
	\definecolor{mycolor1}{rgb}{0.00000,0.44700,0.74100}%
\definecolor{mycolor2}{rgb}{0.85000,0.32500,0.09800}%
\definecolor{mycolor3}{rgb}{0.92900,0.69400,0.12500}%
\definecolor{mycolor4}{rgb}{0.49400,0.18400,0.55600}%
\begin{tikzpicture}

\begin{axis}[%
width=0.387\columnwidth,
height=0.303\columnwidth,
at={(0\columnwidth,0\columnwidth)},
scale only axis,
xmin=1,
xmax=13,
xtick={1,2,3,4,5,6,7,8,9,10,11,12,13},
xticklabels={{1.0e-09},{1.0e-12},{1.0e-15},{1.0e-18},{1.0e-21},{1.0e-24},{1.0e-27},{1.0e-30},{1.0e-33},{1.0e-36},{1.0e-39},{1.0e-42},{1.0e-45},{}},
xticklabel style={rotate=90},
ymode=log,
ymin=1e-45,
ymax=1e-09,
yminorticks=true,
ylabel style={font=\color{white!15!black}},
axis background/.style={fill=white},
title style={font=\bfseries},
title={$\beta = 1.00$},
legend style={legend cell align=left, align=left, draw=none,font=\footnotesize,fill=none}
]
\addplot [color=mycolor1, line width=2.0pt]
  table[row sep=crcr]{%
1	5.1073789093672e-19\\
2	5.10751246708677e-19\\
3	2.03380163937255e-16\\
4	2.03380163937255e-16\\
5	5.08450587190156e-21\\
6	1.271126911365e-25\\
7	1.271126911365e-25\\
8	3.17781838694165e-30\\
9	7.94454873881489e-35\\
10	7.94454873881489e-35\\
11	1.98613787760339e-39\\
12	4.9653464263438e-44\\
13	4.9653464263438e-44\\
};
\addlegendentry{$\|\boldsymbol{\mu} - \widehat{\boldsymbol{\mu}}\|_2$ ep}

\addplot [color=mycolor2, line width=2.0pt]
  table[row sep=crcr]{%
1	4.95668339223654e-20\\
2	4.95681300927539e-20\\
3	1.97379340516374e-17\\
4	1.97379340516374e-17\\
5	4.93448523405251e-22\\
6	1.23362173882032e-26\\
7	1.23362173882032e-26\\
8	3.08405542287231e-31\\
9	7.71014124686834e-36\\
10	7.71014124686834e-36\\
11	1.92753598417243e-40\\
12	4.81884164165308e-45\\
13	4.81884164165308e-45\\
};
\addlegendentry{$\nicefrac{\|\boldsymbol{\mu} - \widehat{\boldsymbol{\mu}}\|_2}{\|\widehat{\boldsymbol{\mu}}\|_2}$ ep}

\addplot [color=mycolor3, line width=2.0pt]
  table[row sep=crcr]{%
1	2.69214776093699e-15\\
2	1.23629203826026e-15\\
3	1.50597978157424e-15\\
};
\addlegendentry{$\|\boldsymbol{\mu} - \widehat{\boldsymbol{\mu}}\|_2$ dp}

\addplot [color=mycolor4, dashed, line width=2.0pt]
  table[row sep=crcr]{%
1	1e-09\\
2	1e-12\\
3	1e-15\\
4	1e-18\\
5	1e-21\\
6	1e-24\\
7	1e-27\\
8	1e-30\\
9	1e-33\\
10	1e-36\\
11	1e-39\\
12	1e-42\\
13	1e-45\\
};
\addlegendentry{Tol.}

\end{axis}

\begin{axis}[%
width=0.387\columnwidth,
height=0.303\columnwidth,
at={(0.478\columnwidth,0\columnwidth)},
scale only axis,
xmin=1,
xmax=13,
xtick={1,2,3,4,5,6,7,8,9,10,11,12,13},
xticklabels={{1.0e-09},{1.0e-12},{1.0e-15},{1.0e-18},{1.0e-21},{1.0e-24},{1.0e-27},{1.0e-30},{1.0e-33},{1.0e-36},{1.0e-39},{1.0e-42},{1.0e-45},{}},
xticklabel style={rotate=90},
ymode=log,
ymin=1.38726056759133e-48,
ymax=1e-09,
yminorticks=true,
ylabel style={font=\color{white!15!black}},
axis background/.style={fill=white},
title style={font=\bfseries},
title={$\beta = 0.50$}
]
\addplot [color=mycolor1, line width=2.0pt, forget plot]
  table[row sep=crcr]{%
1	1.36580537305322e-13\\
2	6.62807747510596e-16\\
3	2.14292275385615e-20\\
4	2.14292275385616e-20\\
5	6.76585834269513e-25\\
6	2.06746544334149e-29\\
7	2.06746544334149e-29\\
8	6.14103958315671e-34\\
9	1.78315123995472e-38\\
10	1.78315123995472e-38\\
11	5.08524087693168e-43\\
12	1.42943674308715e-47\\
13	1.42943674308715e-47\\
};
\addplot [color=mycolor2, line width=2.0pt, forget plot]
  table[row sep=crcr]{%
1	1.32550666981533e-14\\
2	6.43251306126126e-17\\
3	2.07969485197263e-21\\
4	2.07969485197264e-21\\
5	6.5662286422405e-26\\
6	2.00646394342091e-30\\
7	2.00646394342091e-30\\
8	5.95984544187097e-35\\
9	1.7305385587741e-39\\
10	1.7305385587741e-39\\
11	4.93519855242787e-44\\
12	1.38726056759133e-48\\
13	1.38726056759133e-48\\
};
\addplot [color=mycolor3, line width=2.0pt, forget plot]
  table[row sep=crcr]{%
1	1.69816949879821e-14\\
2	9.42055475210265e-16\\
3	1.15377761183014e-15\\
};
\addplot [color=mycolor4, dashed, line width=2.0pt, forget plot]
  table[row sep=crcr]{%
1	1e-09\\
2	1e-12\\
3	1e-15\\
4	1e-18\\
5	1e-21\\
6	1e-24\\
7	1e-27\\
8	1e-30\\
9	1e-33\\
10	1e-36\\
11	1e-39\\
12	1e-42\\
13	1e-45\\
};
\end{axis}
\end{tikzpicture}%
	\caption{Solution of the optimization problem for the Katz centrality with Karate. The left panel contains the solution of the problem without $\|\cdot\|_1$-constraints ($\beta=1.0$), the right panel contains the case with the added sparsity constraints ($\beta = 0.5$). The curve are denoted with \texttt{ep} for the extended precision and \texttt{dp} for the double precision.}
	\label{fig:katz}
\end{figure}
We then run the optimization in the extended precision with a request on tolerances ranging from $10^{-9}$ to $10^{-45}$. From the two panels of Figure~\ref{fig:katz} we observe that for both objective functions the IPM optimization is able of getting the requested tolerance.

We repeat an analogous experiment for the PageRank problem, again on the same test matrix \texttt{Newman/karate}. For this case we consider a target PageRank vector $\hat{\boldsymbol{\pi}}$ obtained from the original PageRank vector $\boldsymbol{\pi}$ ($\alpha = 0.8$) by swapping the first ten and the last entries.
\begin{figure}[htbp]
	\centering
	\definecolor{mycolor1}{rgb}{0.00000,0.44700,0.74100}%
\definecolor{mycolor2}{rgb}{0.85000,0.32500,0.09800}%
\definecolor{mycolor3}{rgb}{0.92900,0.69400,0.12500}%
\definecolor{mycolor4}{rgb}{0.49400,0.18400,0.55600}%
\begin{tikzpicture}

\begin{axis}[%
width=0.387\columnwidth,
height=0.303\columnwidth,
at={(0\columnwidth,0\columnwidth)},
scale only axis,
xmin=1,
xmax=13,
xtick={1,2,3,4,5,6,7,8,9,10,11,12,13},
xticklabels={{1.0e-09},{1.0e-12},{1.0e-15},{1.0e-18},{1.0e-21},{1.0e-24},{1.0e-27},{1.0e-30},{1.0e-33},{1.0e-36},{1.0e-39},{1.0e-42},{1.0e-45},{}},
xticklabel style={rotate=90},
ymode=log,
ymin=4.09175135520712e-50,
ymax=1e-09,
yminorticks=true,
ylabel style={font=\color{white!15!black}},
axis background/.style={fill=white},
title style={font=\bfseries},
title={$\beta = 1.00$},
legend style={legend cell align=left, align=left, draw=none, fill=none, font=\footnotesize}
]
\addplot [color=mycolor1, line width=2.0pt]
  table[row sep=crcr]{%
1	9.73656969569328e-15\\
2	9.73656996438544e-15\\
3	1.93310438351333e-18\\
4	4.86542877124973e-23\\
5	4.86542877124973e-23\\
6	1.22511754755895e-27\\
7	3.10860157088444e-32\\
8	3.10860157088444e-32\\
9	8.1959206713837e-37\\
10	2.44932419244125e-41\\
11	2.44932419244125e-41\\
12	9.15424769243729e-46\\
13	4.09175135520712e-50\\
};
\addlegendentry{$\|\boldsymbol{\pi} - \widehat{\boldsymbol{\pi}}\|_2$ - ep}

\addplot [color=mycolor2, line width=2.0pt]
  table[row sep=crcr]{%
1	4.62743914836529e-14\\
2	4.62743927606495e-14\\
3	9.18734542217971e-18\\
4	2.3123621843556e-22\\
5	2.3123621843556e-22\\
6	5.82254025607311e-27\\
7	1.47740580670246e-31\\
8	1.47740580670246e-31\\
9	3.89522443293688e-36\\
10	1.16407513214374e-40\\
11	1.16407513214374e-40\\
12	4.35068257813161e-45\\
13	1.9446613127863e-49\\
};
\addlegendentry{$\nicefrac{\|\boldsymbol{\pi} - \widehat{\boldsymbol{\pi}}\|_2}{\|\widehat{\boldsymbol{\pi}}\|_2}$ - ep} 

\addplot [color=mycolor3, line width=2.0pt]
  table[row sep=crcr]{%
1	8.16740789256382e-15\\
2	8.87333339467788e-15\\
};
\addlegendentry{$\|\boldsymbol{\pi} - \widehat{\boldsymbol{\pi}}\|_2$ - dp}

\addplot [color=mycolor4, dashed, line width=2.0pt]
  table[row sep=crcr]{%
1	1e-09\\
2	1e-12\\
3	1e-15\\
4	1e-18\\
5	1e-21\\
6	1e-24\\
7	1e-27\\
8	1e-30\\
9	1e-33\\
10	1e-36\\
11	1e-39\\
12	1e-42\\
13	1e-45\\
};
\addlegendentry{Tol.}

\end{axis}

\begin{axis}[%
width=0.387\columnwidth,
height=0.303\columnwidth,
at={(0.478\columnwidth,0\columnwidth)},
scale only axis,
xmin=1,
xmax=13,
xtick={1,2,3,4,5,6,7,8,9,10,11,12,13},
xticklabels={{1.0e-09},{1.0e-12},{1.0e-15},{1.0e-18},{1.0e-21},{1.0e-24},{1.0e-27},{1.0e-30},{1.0e-33},{1.0e-36},{1.0e-39},{1.0e-42},{1.0e-45},{}},
xticklabel style={rotate=90},
ymode=log,
ymin=7.80143385788937e-49,
ymax=1e-09,
yminorticks=true,
ylabel style={font=\color{white!15!black}},
axis background/.style={fill=white},
title style={font=\bfseries},
title={$\beta = 0.50$}
]
\addplot [color=mycolor1, line width=2.0pt, forget plot]
  table[row sep=crcr]{%
1	3.15642260186732e-14\\
2	1.19563673346427e-17\\
3	1.19563673344153e-17\\
4	5.66120966667656e-22\\
5	2.09608459525354e-26\\
6	2.09608459525354e-26\\
7	6.99401779829477e-31\\
8	2.21729566487439e-35\\
9	2.21729566487439e-35\\
10	6.92268715582234e-40\\
11	2.22394797127796e-44\\
12	2.22394797127796e-44\\
13	7.80143385788937e-49\\
};
\addplot [color=mycolor2, line width=2.0pt, forget plot]
  table[row sep=crcr]{%
1	1.50013341178326e-13\\
2	5.68242861765107e-17\\
3	5.682428617543e-17\\
4	2.6905680396116e-21\\
5	9.96193137574104e-26\\
6	9.96193137574104e-26\\
7	3.32400350181938e-30\\
8	1.05380037157017e-34\\
9	1.05380037157017e-34\\
10	3.29010262935903e-39\\
11	1.05696197201477e-43\\
12	1.05696197201477e-43\\
13	3.70773912945418e-48\\
};
\addplot [color=mycolor3, line width=2.0pt, forget plot]
  table[row sep=crcr]{%
1	1.16634795666023e-14\\
2	3.51126201096079e-16\\
};
\addplot [color=mycolor4, dashed, line width=2.0pt, forget plot]
  table[row sep=crcr]{%
1	1e-09\\
2	1e-12\\
3	1e-15\\
4	1e-18\\
5	1e-21\\
6	1e-24\\
7	1e-27\\
8	1e-30\\
9	1e-33\\
10	1e-36\\
11	1e-39\\
12	1e-42\\
13	1e-45\\
};
\end{axis}

\end{tikzpicture}%
	
	\caption{Solution of the optimization problem for the PageRank with Karate. The left panel contains the solution of the problem without $\|\cdot\|_1$-constraints ($\beta=1.0$), the right panel contains the case with the added sparsity constraints ($\beta = 0.5$). The curve are denoted with \texttt{ep} for the extended precision and \texttt{dp} for the double precision.}
	\label{fig:pagerank}
\end{figure}
The behavior in Figure~\ref{fig:pagerank} is analogous to the results in Figure~\ref{fig:katz}; in all cases there is convergence of the algorithm to the relevant tolerance.

\end{document}